\documentclass[acmsmall, nonacm]{acmart}

\usepackage[preludealign]{prelude}
\acmclean

\input{tikz/tikz_macros}

\newcommand{\ifmanuscript}[2]{\ifdefstring{\acmformat}{manuscript}{#1}{#2}}

\newcommand{\figureMathematica}[4][]{%
  \begin{figure}
    \centering
    \includegraphics[#1]{figures/#2.pdf}
    \ifempty{#4}{}{\begin{justify}
      \footnotesize\ignorespaces#4\unskip
    \end{justify}}
    \caption{\ignorespaces#3\unskip}
    \label{fig:#2}
  \end{figure}}

\newcommand{\showDOI}[1]{\unskip}
\newcommand{\showURL}[1]{\unskip}

\begin{document}

\title{How to Schedule Near-Optimally under Real-World Constraints}

\author{Ziv Scully}
\affiliation{%
  \institution{Carnegie Mellon University}
  \department{Computer Science Department}
  \streetaddress{5000 Forbes Ave}
  \city{Pittsburgh}
  \state{PA}
  \postcode{15213}
  \country{USA}}
\email{zscully@cs.cmu.edu}

\author{Mor Harchol-Balter}
\affiliation{%
  \institution{Carnegie Mellon University}
  \department{Computer Science Department}
  \streetaddress{5000 Forbes Ave}
  \city{Pittsburgh}
  \state{PA}
  \postcode{15213}
  \country{USA}}
\email{harchol@cs.cmu.edu}

\begin{abstract}
  Scheduling is a critical part of practical computer systems,
  and scheduling has also been extensively studied from a theoretical perspective.
  Unfortunately, there is a gap between theory and practice,
  as the optimal scheduling policies presented by theory
  can be difficult or impossible to perfectly implement in practice.
  In this work, we use recent breakthroughs in queueing theory to begin to bridge this gap.
  We show how to translate theoretically optimal policies---%
  which provably minimize mean response time (a.k.a. latency)---%
  into \emph{near-optimal} policies that are easily implemented
  in practical settings.
  Specifically, we handle the following real-world constraints:
  \begin{itemize}
  \item
    We show how to schedule in systems
    where job sizes (a.k.a. running times) are unknown, or only partially known.
    We do so using simple policies that achieve performance very close to
    the much more complicated theoretically optimal policies.
  \item
    We show how to schedule in systems
    that have only a limited number of priority levels available.
    We show how to adapt theoretically optimal policies to this constrained setting
    and determine how many levels we need for near-optimal performance.
  \item
    We show how to schedule in systems
    where job preemption can only happen at specific checkpoints.
    Adding checkpoints allows for smarter scheduling, but each checkpoint incurs time overhead.
    We give a rule of thumb that near-optimally balances this tradeoff.
  \end{itemize}
\end{abstract}

\maketitle

\section{Introduction}
\label{sec:intro}

Some form of scheduling is at the heart of every part of every computer system,
and thus there is an extensive literature on scheduling,
from both theoretical
\citep{pinedo_scheduling_2016, gittins_multi-armed_2011, kleinrock_queueing_1976, harchol-balter_performance_2013, wierman_scheduling_2007}
and practical
\citep{arpaci-dusseau_operating_2018, capozzi_downlink_2012, aas_understanding_2005, zhan_cloud_2015}
perspectives.
Scheduling can be very counterintuitive,
and so insights from queueing theory are invaluable
for guiding the design of practical scheduling policies.
However, there is a gap between the idealized models analyzed in queueing theory
and the realities of real-world computer systems.
The goal of this paper is to use recent breakthroughs in queueing theory
to bridge this gap.

\subsection{Optimality in Theory}
\label{sec:intro:theory}

From a theoretical angle, scheduling to minimize mean response time
has been solved in a variety of settings.
In single-server systems where scheduler knows each job's exact size,
the optimal policy is \emph{Shortest Remaining Processing Time} (SRPT),
which always serves the job of minimal remaining size.
SRPT has also been shown to perform well in some multiserver settings
\citep{grosof_srpt_2018, grosof_load_2019, leonardi_approximating_2007}.

Even when job sizes are unknown or only partially known to the scheduler,
theory provides optimal policies for minimizing mean response time,
provided that the \emph{distribution} of job sizes is known.
In the M/G/1 queue, a standard single-server queueing model,
the general optimal policy is called the \emph{Gittins} policy
\citep{gittins_multi-armed_2011, aalto_gittins_2009, aalto_properties_2011}.
Gittins has the same high-level structure as SRPT:
it assigns each job a numerical priority, called its \emph{rank},
and always serves the job of minimal rank.
A job's rank measures, roughly speaking,
how much service the job needs in order to have a good chance of completing.
For example, in the special case of known job sizes,
a job's rank is its remaining size, and thus Gittins is equivalent to SRPT.
But when job sizes are unknown or only partially known,
computing a job's rank is more complicated.

\subsection{Obstacles in Practice}
\label{sec:intro:practice}

Unfortunately, Gittins is often an impractical choice for computer systems.
This is because most systems have to overcome implementation obstacles
that the theory behind Gittins ignores.
\begin{itemize}
\item
  The Gittins policy is unfortunately typically too complicated to use,
  because determining a job's rank requires solving
  a new complex optimization problem at every moment in time.
  In practice, we prefer a simpler scheduling policy.
\item
  Gittins assumes the scheduler has infinitely many priority levels.
  In practice, the scheduler may have only
  a limited number of priority levels to work with.
\item
  Gittins assumes the scheduler can instantly preempt jobs at any time.
  In practice, the scheduler may only be able to preempt jobs at certain times,
  and preemption may not be instant.
\end{itemize}

\subsection{Contributions}
\label{sec:intro:contributions}

This paper is a guide to scheduling in light of practical concerns.
\begin{quote}
  We present \emph{simple scheduling policies}
  that achieve mean response time close to the theoretical optimum.
\end{quote}
Our policies work even in scenarios
with limitations on priority levels and preemption.
We accommodate settings with known job sizes, unknown job sizes, and in between.

We structure our paper as a guide
that walks through the steps of designing and tuning a scheduling policy.
\begin{itemize}
\item
  Our first step is to choose a simple scheduling policy
  to serve as a ``Gittins substitute'' (\cref{sec:substitute}).
  The choice of Gittins substitute depends on system parameters,
  such as the job size distribution
  and the accuracy of job size information available to the scheduler.
\item
  Next, we show how to adapt our chosen Gittins substitute
  to settings with limited priority levels (\cref{sec:lpl}).
  We will answer questions such as
  how many priority levels are necessary
  and how best to use the priority levels one has.
\item
  Finally, we show how to adapt our chosen Gittins substitute
  to settings with limited preemption (\cref{sec:pc}).
  We will answer questions such as
  how frequently one should place preemption checkpoints
  so as to best trade off between enabling frequent preemption
  while avoiding excessive overhead.
\end{itemize}
\Cref{sec:prior_work} reviews prior work,
and \cref{sec:lessons} summarizes the main takeaways of the paper.

\subsection{Techniques}
\label{sec:intro:techniques}

The analysis in this paper is made possible by a recent theoretical breakthrough,
namely the \emph{SOAP framework} \citep{scully_soap_2018}.
SOAP provides the theoretical basis for analyzing the performance of
a wide variety of scheduling policies,
giving an exact formula for mean response time.
We use these formulas to compare the performance of Gittins,
the theoretically optimal policy,
to the performance of simpler policies.
SOAP also enables us to analyze systems
with limitations on priority levels and preemption.
See \citet[Section~3]{scully_soap_2018} for more details
on modeling systems using SOAP.

\subsection{Other Queueing Systems and Performance Metrics}
\label{sec:intro:other_settings}

This paper focuses on the goal of minimizing mean response time
in single-server systems.
We focus on this goal because it is a practical objective
for which theoretically optimal policy is known under very general conditions
(\cref{sec:intro:theory}).
However, we believe the lessons in this paper are likely to generalize well
to minimizing mean response time in multi-server queueing systems.
\begin{itemize}
\item
  In immediate-dispatch systems,
  which have multiple servers each with their own queue,
  one can focus on one server at a time.
\item
  In central-queue systems,
  which have multiple servers all serving jobs from a single queue,
  there is a growing amount of theoretical evidence
  suggesting that policies with optimal or near-optimal mean response time
  in single-server systems
  also have near-optimal mean response time in multiserver systems
  \citep{kalyanasundaram_minimizing_1997, leonardi_approximating_2007,
    becchetti_nonclairvoyant_2004, scully_optimal_2021, grosof_srpt_2019,
    scully_gittins_2020, grosof_load_2019}
\end{itemize}

In some settings,
a job's \emph{slowdown}, which is response time divided by size,
is more important than its response time.
The slowdown metric is suited for settings in which
a 10-second delay to a 100-second job is hardly noticed,
but a 10-second delay to a 1-second job is a major problem.
There is some theory on optimizing mean slowdown
when job sizes are known \citep{hyytia_minimizing_2012}.
We suspect that the lessons in this paper may apply
to minimizing mean slowdown as well as mean response time.

Another very important set of metrics in computer systems
are those related to the tail of the response time distribution.
One such metric is the 99th percentile of response time, or ``T99''.
While the tail of response time is of great importance in practice,
the theory of optimizing the tail is lacking.
It would be valuable to revisit and reevaluate the recommendations in this paper
when the theory of optimizing the tail has matured.

\section{What is Scheduling?}
\label{sec:concepts}

In this section,
we define terminology for queueing systems and scheduling
that we use throughout the paper.

\subsection{Jobs, Servers, and Queues}
\label{sec:concepts:queueing_system}

  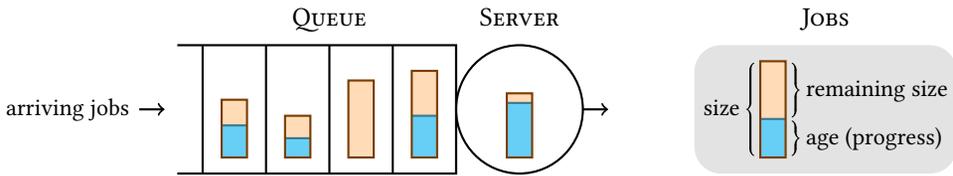
\begin{figure}
    \centering
    \begin{tikzpicture}[figure, x=12, y=12]
    \draw (2, 0) circle (2);
    \node[above] at (2, 2) {\normalsize\strut\textsc{Server}};
    \draw[thick, ->] (4, 0) -- ++(0.8, 0);
    \draw (-8.8, -2) -- ++(8.8, 0) -- ++(0, 4) -- ++(-8.8, 0);
    \node[above] at (-4, 2) {\normalsize\strut\textsc{Queue}};
    \draw (0, -2) -- ++(0, 4);
    \draw (-2, -2) -- ++(0, 4);
    \draw (-4, -2) -- ++(0, 4);
    \draw (-6, -2) -- ++(0, 4);
    \draw (-8, -2) -- ++(0, 4);
    \draw[thick, ->] (-10, 0) node [left] {\strut arriving jobs} -- ++(0.8, 0);
    \newcommand{\job}[3]{%
      \fill[cyan!51!white] (#1 - 0.4, -1.5) rectangle ++(0.8, #3);
      \fill[orange!28!white] (#1 - 0.4, -1.5 + #3) rectangle ++(0.8, #2 - #3);
      \ifstrequal{#3}{0}{}{\draw[cyan!48!black] (#1 - 0.4, -1.5 + #3) -- ++(0.8, 0);}
      \draw[orange!48!black] (#1 - 0.4, -1.5) rectangle ++(0.8, #2);
    }
    \job{2}{2}{1.7}
    \job{-1}{2.7}{1.3}
    \job{-3}{2.4}{0}
    \job{-5}{1.3}{0.6}
    \job{-7}{1.8}{1}
    \begin{scope}[shift={(10, 0)}]
      \fill[black!11, rounded corners=1em] (-2.48, -2) rectangle (5.8, 2);

      \job{0}{3}{1.2}
      \node[above] at (1.66, 2) {\normalsize\strut\textsc{Jobs}};
      \draw[decorate, decoration={calligraphic brace, raise=0.25em}]
        (-0.4, -1.5 + 1/16) -- ++(0, 3 - 1/8)
        node [black, midway, left=0.4em] {\strut size};
      \draw[decorate, decoration={calligraphic brace, mirror, raise=0.25em}]
        (0.4, -1.5 + 1/16) -- ++(0, 1.2 - 1/8)
        node [black, midway, right=0.4em] {\strut age (progress)};
      \draw[decorate, decoration={calligraphic brace, mirror, raise=0.25em}]
        (0.4, -1.5 + 1.2 + 1/16) -- ++(0, 3 - 1.2 - 1/8)
        node [black, midway, right=0.4em] {\strut remaining size};
    \end{scope}
\end{tikzpicture}

    \ifempty{}{}{\begin{justify}
      \footnotesize\ignorespaces\unskip
    \end{justify}}
    \caption{\ignorespaces
  A Single-Server Queueing System
\unskip}
    \label{fig:queue}
  \end{figure}

\Cref{fig:queue} illustrates a queueing system.
\emph{Jobs} arrive to the system over time,
each of which carries some amount of work.
The \emph{server} performs this work, serving one job at a time,
and the \emph{queue} holds jobs waiting for service.

A server need not correspond to a single physical server or machine.
Depending on the scale being considered,
a server might be a single core,
a multicore chip,
a single machine,
or even a cluster of machines.
As other examples, a server might be
a virtual machine instance,
a worker thread in a database,
or a top-of-rack network switch.

In this work, we focus on single-server systems,
but we expect many of our lessons to hold in multiserver systems
(\cref{sec:intro:other_settings}).

\subsection{Job Sizes}
\label{sec:concepts:size}

A job's \emph{size} is amount the work associated with the job.
We can think of a job's size as
the amount of time it would take to complete the job's work
if it were served in isolation.
We thus measure size in units of time.

At any given time, the server may have already made some progress on a job.
Analogously to size,
a job's \emph{remaining size} is
the amount of time it would take to complete the job's \emph{remaining} work
if it were served in isolation.
Upon arrival, a job's remaining size is its size,
but remaining size decreases at rate~$1$ during service.
The job completes when its remaining size reaches~$0$.

\subsection{Stochastic Arrivals and Load}
\label{sec:concepts:load}

We study a \emph{stochastic} queueing model,
assume that jobs arrive to the system at random times,
and that each job has a random size.
\begin{itemize}
\item
  The \emph{arrival rate} of jobs, denoted~$\lambda$,
  is the average number of jobs per unit of time that arrive.
\item
  The \emph{job size distribution}, denoted~$S$,
  is the statistical distribution of job sizes.
  For example, the average job size is $\E{S}$.
\end{itemize}
We specifically study the \emph{M/G/1 queueing model},
which makes certain independence assumptions on arrival times and job sizes.\footnote{%
  In the M/G/1,
  the arrival times are the increments of a Poisson process of rate~$\lambda$,
  and each job's size is drawn from distribution~$S$
  independently of arrival times and other jobs' sizes.}

\figureMathematica[height=1.6in]{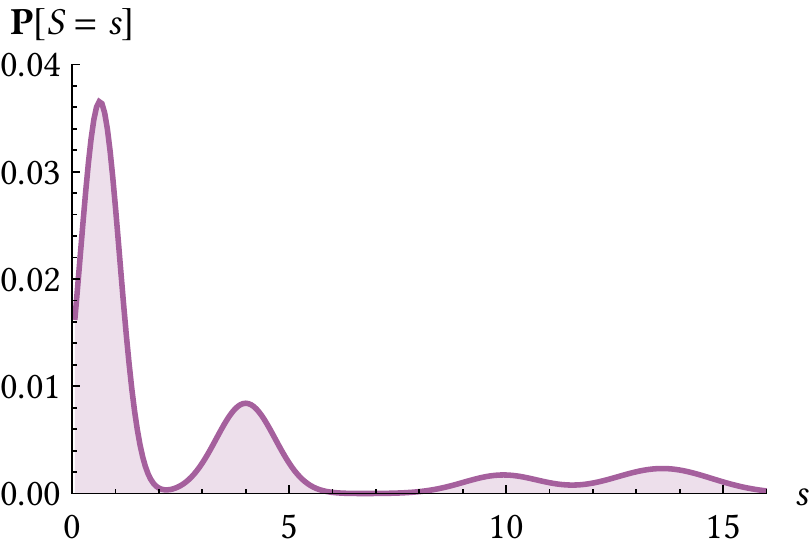}{
  Example Job Size Distribution
}{
  An example of a possible job size distribution~$S$.
  For each possible size~$s$,
  we show the probability that an arriving job's size is~$s$.
  For simplicity, in this example,
  the possible sizes are discretized in increments of $1/16$.
}

Given the arrival rate and job size distribution,
one can compute the \emph{load}, denoted~$\rho$,
which is the fraction of time the server is busy:
\begin{align*}
  \text{load} = \rho = \lambda \E{S}.
\end{align*}
This formula holds because every unit of time,
$\lambda$ jobs arrive on average,
each of which will keep the server busy for $\E{S}$ time on average.
It turns out we need $\rho < 1$ for the system to be stable,
or else work arrives too fast for the server to keep up.
We therefore assume $\rho < 1$ throughout this paper.

\subsection{What Information does the Scheduler Have?}
\label{sec:concepts:settings}

Depending on the setting,
the scheduler may or may not know a job's size or remaining size.
We consider three different settings,
each of which we describe in more detail below.
\begin{itemize}
\item
  \emph{Size-aware:}
  when a job arrives, the scheduler learns its size.
\item
  \emph{Size-oblivious:}
  when a job arrives, the scheduler does not learn any information it.
\item
  \emph{Class-aware:}
  when a job arrives, the scheduler learns some information about it,
  which we call the job's \emph{class},
  but this information does not completely determine its size.
\end{itemize}

\subsubsection{The size-aware setting}
\label{sec:concepts:settings:size-aware}

By definition, the scheduler knows each job's size in this setting.
We also assume that the server can measure
how much time a job has been served so far,
which we call the job's \emph{age}.
The scheduler can compute remaining size from size and age:
\begin{align*}
  \text{remaining size} = \text{size} - \text{age}.
\end{align*}

\subsubsection{The size-oblivious setting}
\label{sec:concepts:settings:size-oblivious}

It may seem like the scheduler
has no useful information at all in this setting.
However, there are two important types of information available.

First, we assume that the scheduler knows each job's age.
This is a reasonable assumption because measuring time served so far

Second, we assume that the scheduler knows the \emph{job size distribution}.
This is a reasonable assumption because even when sizes are unknown,
at the moment a job completes, its age and size are equal,
so we learn job sizes after the fact.
We can thus infer the size distribution from past data.
We denote the size distribution by the random variable~$S$.
Knowing the size distribution~$S$ amounts to knowing $\P{S > x}$,
the probability a job has size greater than~$x$,
for all $x \geq 0$.

We assume the size distribution is also known to the scheduler
in the size-aware and class-aware settings,
but its role in those settings is less important.

\subsubsection{The class-aware setting}
\label{sec:concepts:settings:class-aware}

This setting is an intermediate between the size-aware and size-oblivious settings.
We assume that there are a number of \emph{job classes}.
Each job has one class, and the scheduler knows each job's class.
A job's class models any ``static'' information known about the job:
the user that submitted the job, a priority level, a size estimate, or similar.
One can in theory have infinitely many job classes,
but for simplicity, we only consider systems with a small finite number of classes.

Just like the size-oblivious setting, in the class-aware setting,
the scheduler knows job ages and the overall job size distribution~$S$.
In addition, for each class~$k$,
the scheduler knows the \emph{class-specific job size distribution},
which we denote by~$S_k$.
One can think of $S_k$ as the random variable representing
a job's size conditional on knowing that its class is~$k$.

\subsection{Scheduling: Deciding Which Job to Serve}
\label{sec:concepts:scheduling}

A \emph{scheduling policy} is an algorithm that the scheduler uses
to decide which job to serve at every moment in time.
We can classify scheduling policies into two broad categories:
\begin{itemize}
\item
  \emph{Nonpreemptive} policies never preempt, meaning interrupt,
  jobs once they start service.
\item
  \emph{Preemptive} policies sometimes preempt jobs during service,
  pausing them temporarily to switch to serving a different job.
\end{itemize}

Two examples of scheduling policies
are \emph{First Come, First Serve} (FCFS)
and \emph{Shortest Remaining Processing Time} (SRPT).
FCFS is nonpreemptive, while SRPT is preemptive.
We generally assume that preemptions occur instantly
and that preempted jobs keep their progress toward completion.

\subsubsection{Describing scheduling policies with rank functions}
\label{sec:concepts:scheduling:rank}

The space of scheduling policies is vast.
To narrow our focus,
we consider a certain class of dynamic-priority scheduling policies
\citep{scully_soap_2018}.

\begin{itemize}
\item
  A job's \emph{rank} is a non-negative real number representing a job's priority,
  where lower rank means better priority.
\item
  A \emph{rank function} is a function that determines a job's rank
  using only its age and, if available, size or class:
  \begin{align*}
    \text{rank function} : \text{age, size (if known), class (if known)} \to \text{rank}.
  \end{align*}
\end{itemize}
Nearly every scheduling policy mentioned in this paper
uses a rank function to assign each job a rank,
then follows one core scheduling rule:
\begin{quote}
  \emph{\textbf{Always serve the job of minimum rank.}
  When necessary, break ties in FCFS order.}
\end{quote}

  \begin{figure}
    \centering
    \begin{tikzpicture}[figure]
  \begin{scope}[shift={(-15,0)}]
    \node at (6/2, 9.5) {\normalsize\strut\textsc{FCFS}};

    \yguide[$c$]{0}{2}

    \draw[primary]
    (0, 2) -- (6, 2);

    \axes{6}{6}{$0$}{age~$a$}{$0$}{rank $r(a) = c$}
  \end{scope}

  \begin{scope}[shift={(15,0)}]
    \node at (6/2, 9.5) {\normalsize\strut\textsc{SRPT}};

    \yguide[$s = 22$]{0}{5.5}
    \yguide[$s = 15$]{0}{3.75}
    \yguide[$s = 6$]{0}{1.5}

    \draw[primary]
    (5.5, 0) -- (0, 5.5);
    \draw[primary]
    (3.75, 0) -- (0, 3.75);
    \draw[primary]
    (1.5, 0) -- (0, 1.5);

    \axes{6}{6}{$0$}{age~$a$}{$0$}{rank $r(a, s) = s - a$}
  \end{scope}

  \begin{scope}[shift={(0,0)}]
    \node at (6/2, 9.5) {\normalsize\strut\textsc{FB}};

    \draw[primary]
    (0, 0) -- (6, 6);

    \axes{6}{6}{$0$}{age~$a$}{$0$}{rank $r(a) = a$}
  \end{scope}

  \node at (-21.5 + 6/2, 0) {};
  \node at (21.5 + 6/2, 0) {};
\end{tikzpicture}

    \ifempty{}{}{\begin{justify}
      \footnotesize\ignorespaces\unskip
    \end{justify}}
    \caption{\ignorespaces
  Examples of Rank Functions
\unskip}
    \label{fig:rank_examples}
  \end{figure}
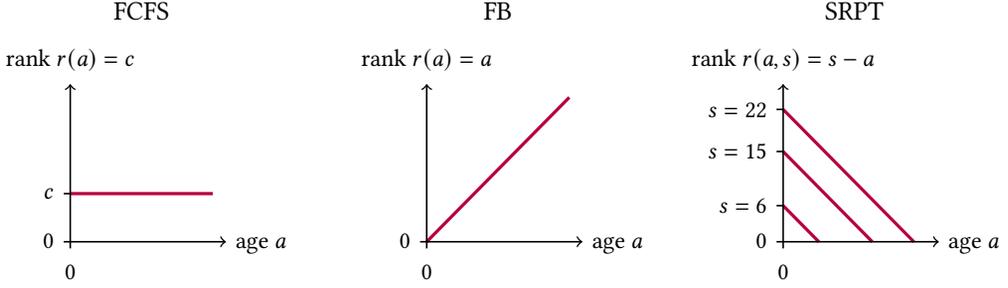

We can express a wide variety of scheduling policies in this framework
by simply varying the rank function
\citep{scully_soap_2018}.
\Cref{fig:rank_examples} illustrates the rank functions of the following policies:
\begin{itemize}
\item
  Under FCFS, a job's rank is a fixed constant~$c$
  at all ages~$a$:
  \begin{align*}
    \text{rank}_{\text{FCFS}}(a) = c.
  \end{align*}
  This means tiebreaking rule is always in effect,
  so jobs are served in FCFS order.
\item
  Under a policy called \emph{foreground-background} (FB), a job's rank is its age~$a$:
  \begin{align*}
    \text{rank}_{\text{FB}}(a) = a.
  \end{align*}
  This means the scheduler always serves the job of least age,
  namely the job that has received the least service so far.\footnote{%
    When multiple jobs are tied for least age,
    FB shares the server between them.
    This is a natural consequence of always serving the job of least rank,
    which in FB's case is least age.
    For example, when two jobs $J$ and~$K$ are tied for least age,
    if FB serves $J$ for an instant, then $J$'s age increases,
    so FB switches to serving~$K$.}
\item
  Under SRPT, a job's rank is its size~$s$ minus its age~$a$:
  \begin{align*}
    \text{rank}_{\text{SRPT}}(a, s) = s - a.
  \end{align*}
  This means the scheduler always serves the job of least remaining size.
\end{itemize}

\subsection{Response Time}

A job's \emph{response time} is the amount of time between
its arrival and its completion.
This work focuses on minimizing \emph{mean response time},
which we write as $\E{T}$.
Mean response time is impacted by three parameters:
\begin{itemize}
\item
  The scheduling policy.
\item
  The job size distribution~$S$,
  including the class-specific job size distributions~$S_k$
  in the class-aware setting.
\item
  The arrival rate~$\lambda$, or equivalently load $\rho = \lambda \E{S}$.
\end{itemize}

Of these three parameters, the simplest one to understand is load~$\rho$,
namely the fraction of time the server is busy.
In the $\rho \to 0$ limit, the system is almost always empty,
so there is almost no queueing and $\E{T} \to \E{S}$.
In the $\rho \to 1$ limit, the system barely keeps up with demand,
and we actually have $\E{T} \to \infty$.
In general, $\E{T}$ is a strictly increasing function of load~$\rho$.\footnote{%
  Here and throughout this paper,
  when we talk about varying the load $\rho = \lambda \E{S}$,
  we mean that the arrival rate~$\lambda$ varies
  while the job size distribution~$S$ remains fixed.}

\figureMathematica{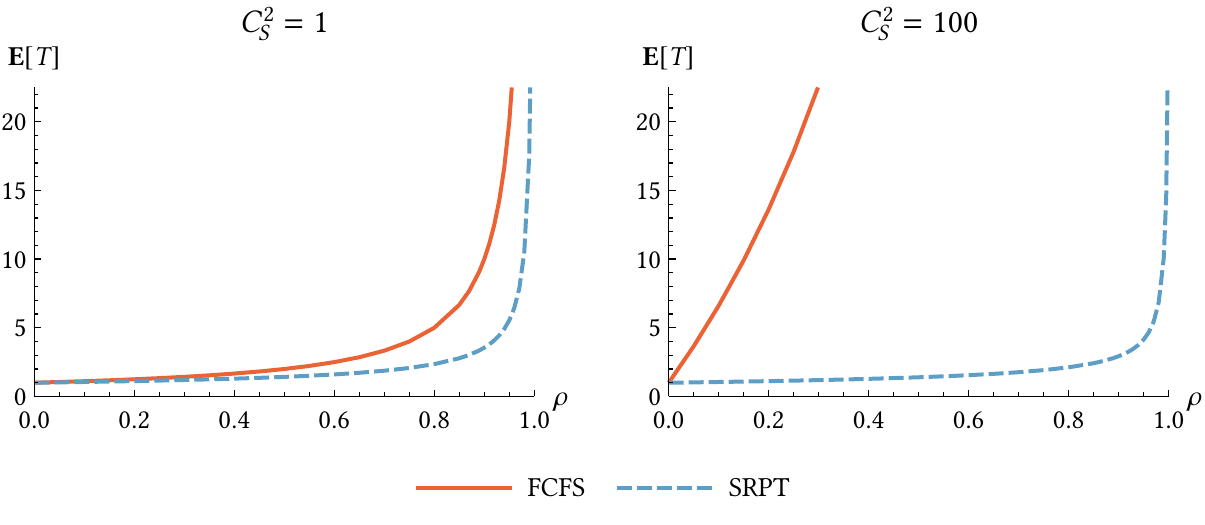}{
  Effect of Variability on Response Time of FCFS and SRPT
}{
  Mean response time of FCFS and SRPT for two job size distributions,
  one with low variability, namely $C^2_S = 1$ (left),
  and one with high variability, namely $C^2_S = 100$ (right).
  The $C^2_S = 1$ job size distribution
  is an exponential distribution with mean~$1$.
  The $C^2_S = 100$ job size distribution
  is a two-phase hyperexponential distribution with mean~$1$
  and ``balanced means'',
  meaning the probabilities and rates of each phase are proportional.
}

The effect of the job size distribution~$S$ on mean response time
is more complicated, as it depends on the scheduling policy.
As a simple example, \cref{fig:fcfs_vs_srpt} shows how
the squared coefficient of variation $C^2_S = \Var{S}/\E{S}^2$,
which measures the variability of~$S$
(greater $C^2_S$ means more variable),
affects the mean response time of FCFS and SRPT.
\begin{itemize}
\item
  FCFS's mean response time increases with~$C^2_S$.
  In fact, one term of FCFS's mean response time formula
  is directly proportional to $C^2_S$.
  The reason for this is the so-called ``inspection paradox'',
  which states that when a job arrives,
  the job currently in service is disproportionately likely to be a large job.
  This means that for high-variability job size distributions,
  under FCFS,
  many small jobs can get stuck behind one very large job,
  which increases mean response time.
\item
  SRPT's mean response time is much less sensitive to~$C^2_S$.
  In fact, SRPT's mean response time is slightly \emph{lower}
  in the high-variability case.
  This is because when SRPT is given a greater range of job sizes,
  there are more opportunities to prioritize very small jobs ahead of larger jobs.
\end{itemize}

In general,
the impact of the scheduling policy and job size distribution on mean response time
has no simple characterization.
Even for a fixed job size distribution,
whether one policy outperforms another can depend on the load.
See \cref{fig:example_rts} for an example of this.

\section{Creating a Gittins Substitute}
\label{sec:substitute}

The Gittins policy minimizes mean response time.
The form that Gittins takes depends on the setting
\citep{gittins_multi-armed_2011}.
\begin{itemize}
\item
  In the size-aware setting, Gittins reduces to SRPT.
\item
  In the size-oblivious setting,
  Gittins uses the following rank function,
  which takes into account the job size distribution~$S$:
  \begin{equation}
    \label{eq:gittins_rank}
    \text{rank}_{\text{Gittins}}(a)
    = \inf_{b > a} \frac{\E{\min\{S - a, b\} \given S > a}}{\P{S \leq b \given S > a}}.
  \end{equation}
\item
  In the class-aware setting,
  Gittins has a rank function similar to~\cref{eq:gittins_rank} with one change:
  when computing the rank for a job of class~$k$,
  we use the class-specific job size distribution~$S_k$ in place of~$S$
  (\cref{sec:substitute:class-aware}).
\end{itemize}
Looking at a high level,
the Gittins rank function trades off between two ideas.
The first is that we want to favor jobs whose expected remaining size is small,
which is captured by the numerator of~\cref{eq:gittins_rank}.
The second is that we want to favor jobs that are likely to complete soon,
which is captured by the denominator of~\cref{eq:gittins_rank}.
\Cref{fig:hills_size-oblivious} shows what the Gittins rank function looks like
for the job size distribution from \cref{fig:example_dist}.

  \begin{figure}
    \centering
    \begin{tikzpicture}[figure]
  \begin{scope}[shift={(-10,0)}]
    \node at (13/2, 9) {\normalsize\strut\textsc{Gittins}};

    \xguide[$5$]{5/16 * 13}{0}
    \xguide[$10$]{10/16 * 13}{0}
    \xguide[$15$]{15/16 * 13}{0}

    \yguide[$2$]{0}{2/16 * 13}
    \yguide[$4$]{0}{4/16 * 13}
    \yguide[$6$]{0}{6/16 * 13}

    \begin{scope}[scale=13]
      \draw[primary] \gittinsNoInfo;
    \end{scope}

    \axes{13}{5.5}{$0$}{age}{$0$}{rank}
  \end{scope}

  \begin{scope}[shift={(10,0)}]
    \node at (13/2, 9) {\normalsize\strut\textsc{SERPT}};

    \xguide[$5$]{5/16 * 13}{0}
    \xguide[$10$]{10/16 * 13}{0}
    \xguide[$15$]{15/16 * 13}{0}

    \yguide[$2$]{0}{2/16 * 13}
    \yguide[$4$]{0}{4/16 * 13}
    \yguide[$6$]{0}{6/16 * 13}

    \begin{scope}[scale=13]
      \draw[primary] \serptNoInfo;
    \end{scope}

    \axes{13}{5.5}{$0$}{age}{$0$}{rank}
  \end{scope}

  \node at (-21.5 + 13/2, 0) {};
  \node at (21.5 + 13/2, 0) {};
\end{tikzpicture}

    \ifempty{
  We show the rank functions of Gittins and SERPT
  for the job size distribution shown in \cref{fig:example_dist}.
  Similarly to the distribution,
  both rank functions are discretized in age increments of~$1/16$.
}{}{\begin{justify}
      \footnotesize\ignorespaces
  We show the rank functions of Gittins and SERPT
  for the job size distribution shown in \cref{fig:example_dist}.
  Similarly to the distribution,
  both rank functions are discretized in age increments of~$1/16$.
\unskip
    \end{justify}}
    \caption{\ignorespaces
  Rank Functions of Gittins and SERPT in the Size-Oblivious Setting
\unskip}
    \label{fig:hills_size-oblivious}
  \end{figure}

As shown by \cref{fig:example_rts},
Gittins can in general perform significantly better
than simple size-oblivious policies like FCFS and FB.
Unfortunately, Gittins is a complicated policy.
Implementing Gittins requires solving an optimization problem at every age~$a$,
and these optimization problems are hard to solve in general
\citep[Appendix~B]{scully_simple_2020}.

Fortunately, most of the time,
there are simpler policies that perform nearly as well as Gittins.
The main obstacle is that, like Gittins,
we need to adapt the policy to the job size distribution
and the information known by the scheduler.
In the remainder of this section,
we provide simple Gittins substitutes for the
size-aware (\cref{sec:substitute:size-aware}),
size-oblivious (\cref{sec:substitute:size-oblivious}),
and class-aware (\cref{sec:substitute:class-aware}) settings.

\subsection{Near-Optimal Scheduling in the Size-Aware Setting}
\label{sec:substitute:size-aware}

In the size-aware setting,
Gittins actually takes a reasonably simple form:
it is equivalent to SRPT, which we know always minimizes mean response time
\citep{schrage_proof_1968}.

\subsubsection{When are nonpreemptive policies good enough?}
\label{sec:substitute:size-aware:nonpreemptive}

With this said, there is one aspect of SRPT
that can be an obstacle to implementing it in some systems:
it is a preemptive policy.
This prompts a question: can we achieve performance comparable to SRPT
without any preemptions?

For job size distributions~$S$ with
low squared coefficient of variation $C_S^2 = \Var{S}/\E{S}^2$,
a nonpreemptive version of SRPT
called \emph{Shortest Job First}~(SJF)
has mean response time comparable to that of SRPT.
SJF never preempts a job during service,
but whenever a job completes,
SJF selects the shortest job to serve next.

In fact, for distributions with especially low variance,
even FCFS can be competitive with SRPT.
For example, if the minimum possible job size is~$s_{\min}$
and the maximum job size is~$s_{\max}$,
then FCFS's mean response time is at most $s_{\max}/s_{\min}$ times SRPT's.
We prove this interesting fact in \cref{sec:fcfs_approximation}.
In practice, FCFS still performs well
when most but not all jobs fall between $s_{\min}$ and~$s_{\max}$,
such as when $S$ is Gaussian.

\subsection{Near-Optimal Scheduling in the Size-Oblivious Setting}
\label{sec:substitute:size-oblivious}

\figureMathematica{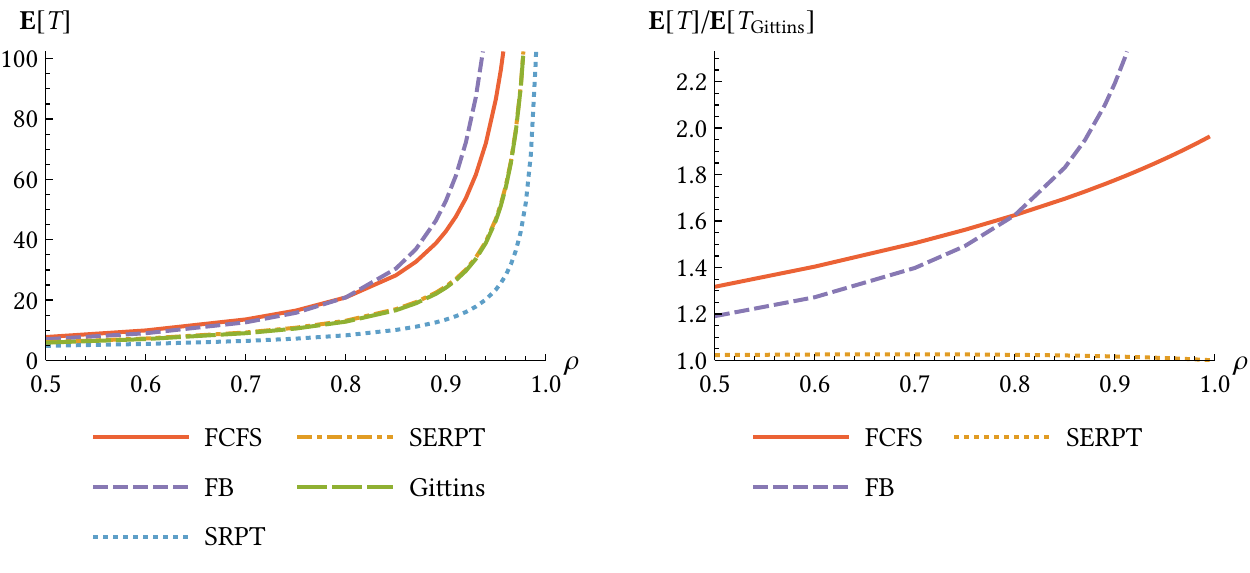}{
  Response Times of Different Scheduling Policies for Example Job Size Distribution
}{
  Mean response times of several scheduling policies
  for the job size distribution from \cref{fig:example_dist}.
  We show both mean response times (left)
  and the mean response time ratios relative to Gittins (right),
  which is the policy that minimizes mean response time
  in the size-oblivious setting.
  In this example, FB outperforms FCFS at low load,
  but FCFS is the better choice at high load.
  However, neither FCFS nor FB performs nearly as well as Gittins or SERPT.
  As expected, SRPT performs even better than Gittins,
  because it uses job size information to make smarter scheduling decisions.
}

In the size-oblivious setting, as we saw in \cref{eq:gittins_rank},
Gittins is impractically complicated to implement
for general job size distributions~$S$.
We thus recommend two different strategies for creating a Gittins substitute.
Both strategies make use of $S$ in a much simpler way than Gittins.

\subsubsection{SERPT: using expected remaining size}

Suppose a job has reached age~$a$ but hasn't finished yet.
If we do not know the job's size, we cannot compute its remaining size.
However, we can compute its \emph{expected} remaining size using $a$ and~$S$:
\begin{align*}
  \text{expected remaining size of a job at age~$a$} = \E{S - a \given S > a}.
\end{align*}
Inspired by SRPT, which prioritizes the job of least remaining size,
we might think to prioritize the job of least \emph{expected} remaining size.
This is what the \emph{Shortest Expected Remaining Processing Time} (SERPT) policy does.
That is, SERPT's rank function is
\begin{align*}
  \text{rank}_{\text{SERPT}}(a) = \E{S - a \given S > a}.
\end{align*}

SERPT's rank function is thus considerably simpler to compute than Gittins's.
Despite this, SERPT's rank function often looks qualitatively similar to Gittins's.
For example, \cref{fig:hills_size-oblivious} shows that
Gittins and SERPT have similar rank functions
for the job size distribution from \cref{fig:example_dist}.
This suggests that SERPT may be a good Gittins substitute,
and there is theoretical evidence supporting this \citep{scully_simple_2020}.
This is certainly true for the job size distribution from \cref{fig:example_dist}:
in \cref{fig:example_rts},
SERPT's mean response time is within 3\% of Gittins's at all loads.
In other numerical experiments (\cref{sec:substitute:worst}),
we have observed near-optimal performance from SERPT
on a wide variety of job size distributions.\footnote{%
  It is an open problem to prove a general bound on SERPT's mean response time
  relative to Gittins's.
  The worst-case scenario for SERPT that has been found to date
  is a pathological job size distribution constructed by
  \citet[Section~7]{scully_simple_2020}.
  Even in this worst-case scenario,
  SERPT does only twice as badly as Gittins,
  whereas FCFS and FB do even worse.}

\subsubsection{Further simplifying SERPT}
\label{sec:substitute:size-oblivious:further_simplify}

While SERPT is already significantly simpler than Gittins,
it is often possible to get away with even simpler scheduling policies.
We begin by observing that for some job size distributions~$S$,
SERPT already reduces to a simple scheduling policy.
\begin{itemize}
\item
  Suppose a job's expected remaining size is largest before it starts service.
  Then SERPT prioritizes a job that has started service
  over any job that has not yet begun service,
  so SERPT reduces to FCFS.
\item
  Suppose a job's expected remaining size strictly increases as it runs.
  Then SERPT prioritizes the job of least age,
  so SERPT reduces to FB.
\end{itemize}
In these scenarios, there is no need to further simplify SERPT.
However, the conditions on~$S$, particularly the latter,
are somewhat restrictive.
Fortunately, our case studies show that
even if $S$ only approximately satisfies one of the conditions,
FCFS or FB has performance comparable to SERPT and Gittins.

\subsection{Near-Optimal Scheduling in the Class-Aware Setting}
\label{sec:substitute:class-aware}

\figureMathematica{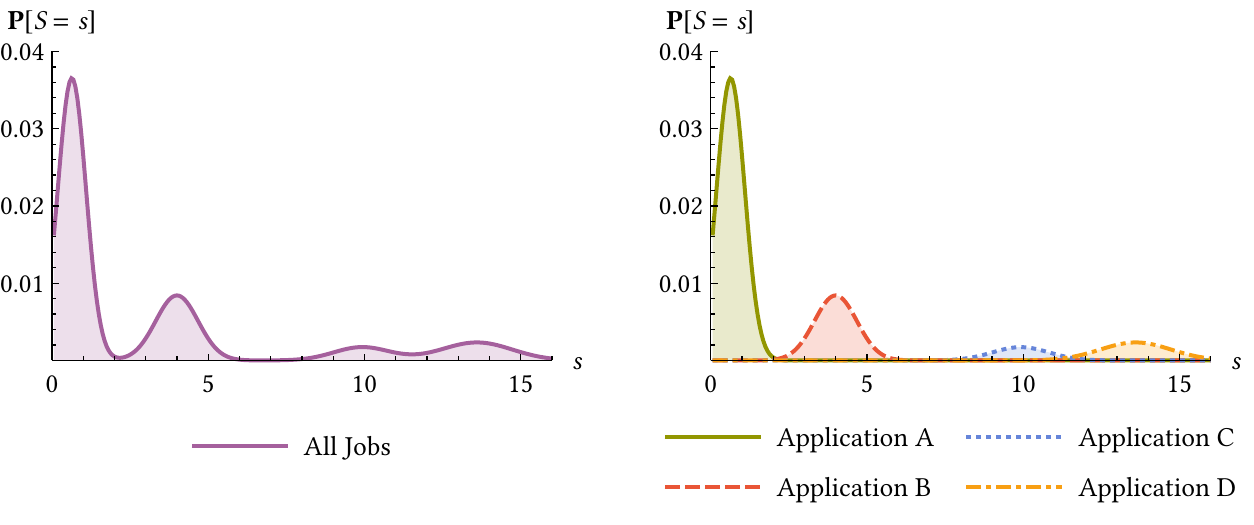}{
  Size Distributions of Jobs from Four Applications
}{}

In the class-aware setting, Gittins has the following rank function,
which takes into account a job's age and class.
The rank of a job of class~$k$ at age~$a$ is
\begin{align*}
  \text{rank}_{\text{Gittins}}(a, k)
  = \inf_{b > a} \frac{\E{\min\{S_k - a, b\} \given S_k > a}}{\P{S_k \leq b \given S_k > a}}.
\end{align*}
This is the same formula as~\cref{eq:gittins_rank},
the rank function for the size-oblivious setting,
but it replaces the job size distribution~$S$
with the \emph{class-specific} job size distribution~$S_k$.
This means that in the class-aware setting,
much like in the size-oblivious setting,
Gittins is impractically complicated to implement in general.
Fortunately, as we will soon see, a class-aware SERPT policy
still serves as a good Gittins substitute.

\subsubsection{Running example: a family of class-aware systems}
\label{sec:substitute:class-aware:example}

To compare different class-aware scheduling policies,
we consider the following running example.
Consider a system serving jobs from four applications,
which we call A, B, C, and~D.
Suppose that each application's job size distribution
is as shown in \cref{fig:sys_dists}.
That is, if we have one job from each application,
it is likely that their size ordering from smallest to largest
is $\text{A} < \text{B} < \text{C} < \text{D}$.

In our running example,
we consider three ways of splitting the four applications into two classes.
\begin{itemize}
\item
  \emph{System~1122:}
  class~1 is A and~B, and class~2 is C and~D.
\item
  \emph{System~1212:}
  class~1 is A and~C, and class~2 is B and~D.
\item
  \emph{System~1221:}
  class~1 is A and~D, and class~2 is B and~C.
\end{itemize}
For example, in System~1212,
the scheduler knows that a class~2 job comes from application B or~D.

\subsubsection{Generalizing SERPT to the class-aware setting}

  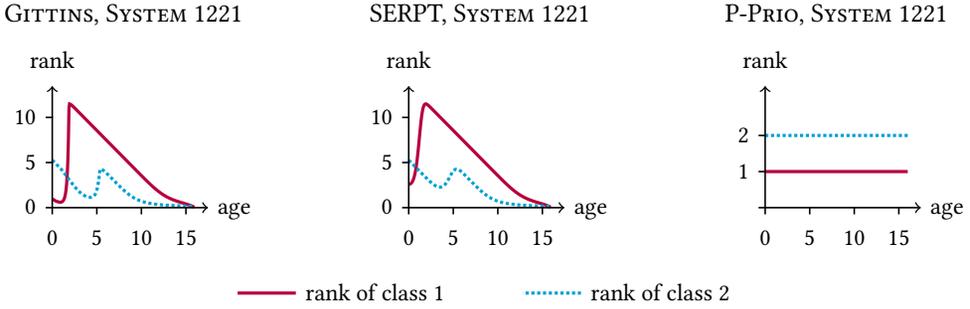
\begin{figure}
    \centering
    \begin{tikzpicture}[figure]
  \begin{scope}[shift={(-15,0)}]
    \node at (6/2, 8) {\normalsize\strut\textsc{Gittins, System~1221}};

    \xguide[$5$]{5/16 * 6}{0}
    \xguide[$10$]{10/16 * 6}{0}
    \xguide[$15$]{15/16 * 6}{0}

    \yguide[$5$]{0}{5/16 * 6}
    \yguide[$10$]{0}{10/16 * 6}

    \begin{scope}[scale=6]
      \draw[secondary] \gittinsAbbaA;
      \draw[tertiary] \gittinsAbbaB;
    \end{scope}

    \axes{6}{4.5}{$0$}{age}{$0$}{rank}
  \end{scope}

  \begin{scope}[shift={(0,0)}]
    \node at (6/2, 8) {\normalsize\strut\textsc{SERPT, System~1221}};

    \xguide[$5$]{5/16 * 6}{0}
    \xguide[$10$]{10/16 * 6}{0}
    \xguide[$15$]{15/16 * 6}{0}

    \yguide[$5$]{0}{5/16 * 6}
    \yguide[$10$]{0}{10/16 * 6}

    \begin{scope}[scale=6]
      \draw[secondary] \serptAbbaA;
      \draw[tertiary] \serptAbbaB;
    \end{scope}

    \axes{6}{4.5}{$0$}{age}{$0$}{rank}
  \end{scope}

  \begin{scope}[shift={(15,0)}]
    \node at (6/2, 8) {\normalsize\strut\textsc{P-Prio, System~1221}};

    \xguide[$5$]{5/16 * 6}{0}
    \xguide[$10$]{10/16 * 6}{0}
    \xguide[$15$]{15/16 * 6}{0}

    \yguide[$1$]{0}{1.5}
    \yguide[$2$]{0}{3}

    \draw[secondary] (0, 1.5) -- (6, 1.5);
    \draw[tertiary] (0, 3) -- (6, 3);

    \axes{6}{4.5}{}{age}{$0$}{rank}
  \end{scope}

  \draw[secondary] (-6 - 1.75 + 6/2, -3.55) -- ++(-2.45, 0);
  \draw[tertiary] (6 - 1.75 + 6/2, -3.55) -- ++(-2.45, 0);
  \node[right] at (-6 - 1.75 + 6/2, -3.5) {rank of class~1};
  \node[right] at (6 - 1.75 + 6/2, -3.5) {rank of class~2};

  \node at (-21.5 + 6/2, 0) {};
  \node at (21.5 + 6/2, 0) {};
\end{tikzpicture}

    \ifempty{}{}{\begin{justify}
      \footnotesize\ignorespaces\unskip
    \end{justify}}
    \caption{\ignorespaces
  Rank Functions in the Class-Aware Setting, System 1221
\unskip}
    \label{fig:hills_class-aware}
  \end{figure}

We have seen that in the size-oblivious setting,
SERPT can serve as a substitute for Gittins.
Recall that SERPT always serves the job of least expected remaining size.
We can use SERPT in the class-aware setting.
The rank of a job of class~$k$ at age~$a$ is
\begin{align*}
  \text{rank}_{\text{SERPT}}(a, k)
  &= \text{expected remaining size of a class~$k$ job at age~$a$} \\
  &= \E{S_k - a \given S_k > a}.
\end{align*}

It turns out that SERPT is a good Gittins substitute in the class-aware setting.
In all of the class-aware systems
described in \cref{sec:substitute:class-aware:example},
SERPT maintains mean response time within 12\% of Gittins's,
as shown by \cref{fig:sys_comparison}.
This is because, as in the size-oblivious setting,
SERPT's rank function tends to look qualitatively similar to Gittins's,
\Cref{fig:hills_class-aware} demonstrates this for System~1221.
In this case,
both policies initially assign class~1 jobs low rank,
because most class~1 jobs in System~1221 are from application~A.
However, once a class~1 job has run for long enough,
it becomes most likely that it is from application~D,
meaning it is likely to be large,
so both policies' rank functions increase accordingly.

\subsubsection{Can we further simplify class-aware SERPT?\nopunct}
\label{sec:substitute:class-aware:further_simplify}

\newcommand{\figSysComparison}{\figureMathematica{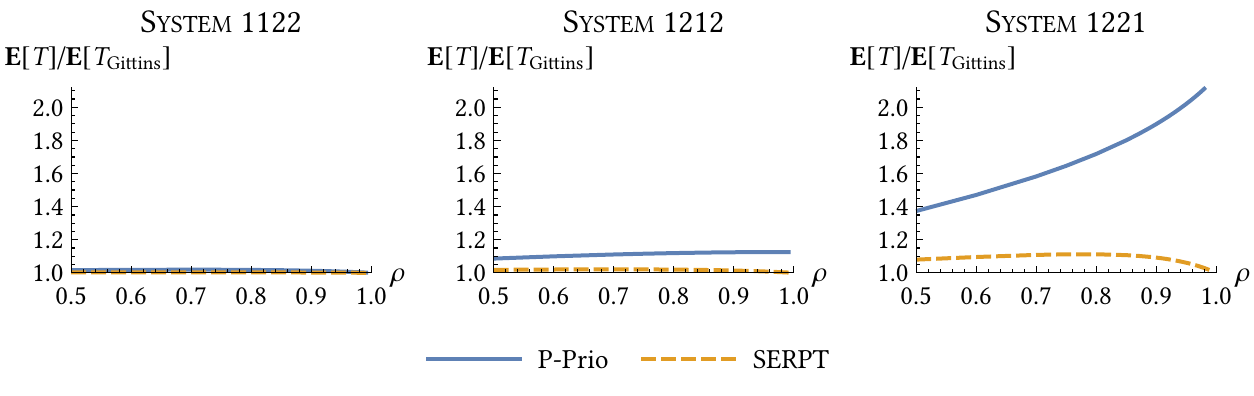}{
  Response Time Comparison in Class-Aware Setting
}{}}

\ifmanuscript{\figSysComparison}{}

In the size-oblivious setting,
we saw that we could often maintain good performance
with policies even simpler than SERPT,
which itself is already much simpler than Gittins.
It is thus natural to ask
whether we can further simplify SERPT in the class-aware setting.

Perhaps the simplest scheduling policy that uses class information
is the \emph{Preemptive Priority} (P-Prio) policy.
Given an ordering of the classes from best to worst priority,
P-Prio always serves the job in the best possible priority class,
serving jobs within each class in FCFS order.
One can make P-Prio favor short jobs by ordering the classes
in order of expected job size.
Can we use P-Prio as a Gittins substitute?

To answer this question,
we compare P-Prio to SERPT and Gittins
in the class-aware systems
described in \cref{sec:substitute:class-aware:example}.
In all cases, P-Prio prioritizes class~1 over class~2.
P-Prio sometimes performs comparably to SERPT and Gittins,
but sometimes SERPT and Gittins are clearly superior.
\begin{itemize}
\item
  In System~1122,
  all three of P-Prio, SERPT, and Gittins have very similar mean response time.
  This makes sense because in this system,
  the scheduler knows a class~1 job is smaller than a class~2 job
  with high probability,
  so strictly prioritizing class~1 jobs is an excellent heuristic.
\item
  In System~1212,
  P-Prio has worse mean response time than SERPT and Gittins,
  but only by 10--15\%.
  This makes sense because in this system,
  the scheduler knows a class~1 job is more likely than not
  to be smaller than a class~2 job.
  However, exceptions occur frequently enough that
  strictly prioritizing class~1 jobs is not as effective as in System~1122.
\item
  In System~1221, SERPT and Gittins significantly outperform P-Prio.
  In fact, using P-Prio turns out to be even worse than using
  the size-oblivious version of SERPT or Gittins.
  This makes sense because in this system,
  given a class~1 job and a class~2 job,
  it is not clear to the scheduler which is larger.
  P-Prio prioritizes class~1 over class~2 because
  class~1 jobs are smaller on average.
  However, once a class~1 job reaches a large enough age,
  its remaining size is likely to be large.
  As shown in \cref{fig:hills_class-aware},
  SERPT and Gittins deal with this by updating a class~1 job's rank during service,
  but P-Prio is limited by its static priorities.
\end{itemize}

\ifmanuscript{}{\figSysComparison}

\subsection{Additional Examples}
\label{sec:substitute:worst}

Through our discussion of the size-oblivious (\cref{sec:substitute:size-oblivious})
and class-aware (\cref{sec:substitute:class-aware}) settings,
we have focused on the scenario shown in \cref{fig:sys_dists}.
Our primary finding is that SERPT is near-optimal
in the size-oblivious and class-aware settings.
However, one might worry that this finding is specific to
the distribution in \cref{fig:sys_dists}.

To test whether SERPT is near-optimal under broader conditions,
we repeated this section's analysis
for 100 different randomly generated scenarios
with the goal of finding the \emph{worst case scenario} for SERPT.
Each scenario was similar to \cref{fig:sys_dists}
in that the overall job size distribution is a mixture of jobs from four applications,
each with a discretized Gaussian distribution,
but we randomly generated the parameters of each Gaussian.

\Cref{fig:worst} summarizes the results by showing
the \emph{worst-case mean response time ratio} between several policies and Gittins
across all 100 scenarios.
We focus on high load $\rho = 0.95$
to emphasize the differences between the scheduling policies.\footnote{%
  We actually computed the ratio at all loads,
  but \cref{fig:worst} shows only $\rho = 0.95$ for simplicity.
  SERPT's worst-case mean response time ratio increases by no more than $0.01$.}
To clarify what \cref{fig:worst} represents,
consider the SERPT bar in the left chart, which has value $1.072$.
This means that for the size-oblivious setting at load $\rho = 0.95$,
across all 100 generated scenarios,
the maximum ratio $\E{T_{\text{SERPT}}} / \E{T_{\text{Gittins}}}$ was~$1.072$.
The bars for other policies and the class-aware setting are computed similarly.

\figureMathematica{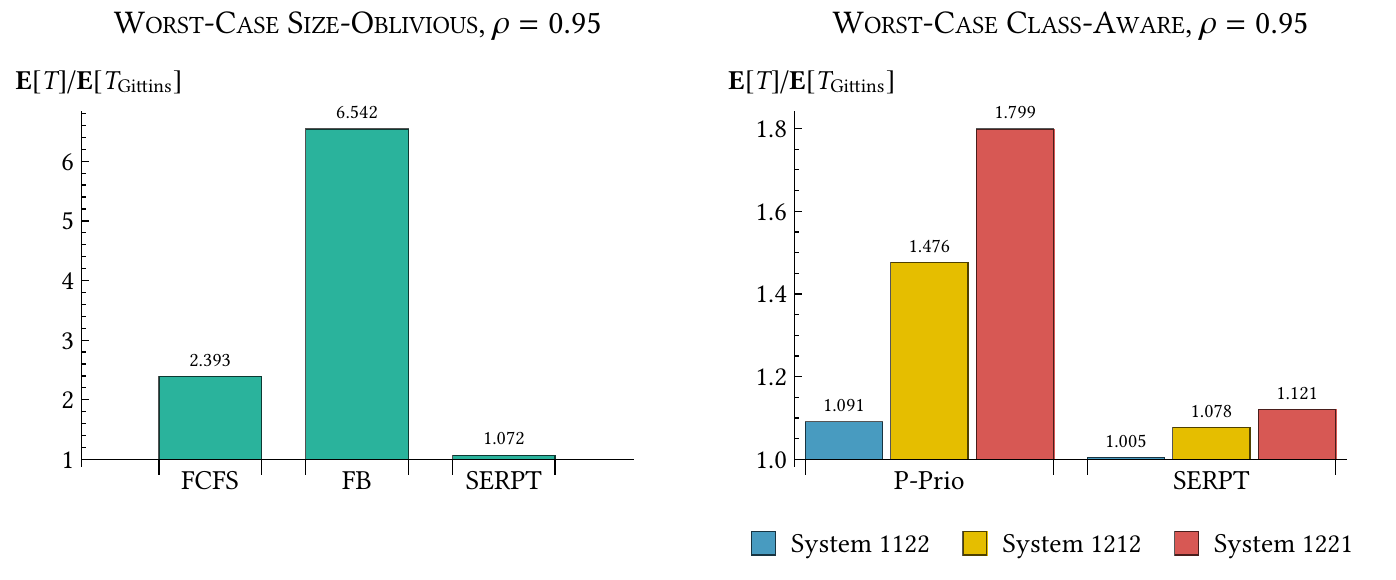}{
  Worst-Case Response Time Ratio out of 100 Randomly Generated Scenarios
}{
  Mean response time ratios relative to Gittins
  in the size-oblivious (left) and class-aware (right) settings
  at load $\rho = 0.95$.
  Each bar represents the worst-case ratio out of 100 randomly generated scenarios,
  each of which is a variation of \cref{fig:sys_dists} with different parameters.
  Specifically, the job size distribution is a mixture
  of four applications' distributions,
  each of which is a Gaussian with uniformly distributed
  mean and standard deviation,
  discretized and restricted to the interval $[0, 16]$.
  We ensure that in each scenario,
  one application's mean is in each of the intervals
  $[0, 4]$, $[4, 8]$, $[8, 12]$, and $[12, 16]$.
  We label the applications in order of increasing mean as A, B, C, and~D,
  from which we define Systems~1122, 1212, and~1221,
  as described in \cref{sec:substitute:class-aware}.
}

The main takeaway of \cref{fig:worst} is that
the conclusions of this section appear to be robust to changing the parameters
of the scenario shown in \cref{fig:sys_dists}.
In particular, SERPT is the only policy that consistently achieves near-optimal performance.

\subsection{Summary: SERPT Suffices}
\label{sec:substitute:summary}

In the size-oblivious and class-aware settings,
the optimal Gittins policy can be impractically complicated,
creating the need for a simpler substitute.
We have seen that SERPT
consistently achieves mean response time close to that of Gittins's,
despite having a much simpler rank function.
We therefore recommend SERPT as a substitute for Gittins for most situations.
In some specific cases, one can use a policy even simpler than SERPT
(\cref{sec:substitute:size-oblivious:further_simplify,
  sec:substitute:class-aware:further_simplify}),
but there are also cases where SERPT outperforms simpler heuristics.

\section{Limited Priority Levels}
\label{sec:lpl}

Many practical computer systems
permit only a small finite number of priority levels, or ranks.
For example,
network switches have a small number of priority levels, typically at most~$8$,
built into their hardware \citep{montazeri_homa_2018}.
Similarly, the Linux kernel packet scheduler has
a configurable number of priority levels, with a default of~$3$
\citep{harchol-balter_size-based_2003}.
We call this the \emph{Limited-Priority-Level} (LPL) setting.

Unfortunately, many scheduling policies assume a continuum of ranks.
For example, under SRPT, a job's rank is its remaining size,
for which there are infinitely many possible values.
System designers who would like to implement policies like SRPT in LPL settings
are thus confronted with the challenge of approximating their policy
using a finite number of possible ranks.
SERPT and Gittins suffer from the same problem.

Given an ``ideal'' policy, such as SRPT,
a common approach to scheduling in the LPL setting is to categorize jobs into levels
based on their rank under the ideal policy
\citep{montazeri_homa_2018, harchol-balter_size-based_2003}.
However, it is not clear how best to do this.
Continuing with the SRPT example, if one has only two priority levels,
we can assign rank~1 to jobs smaller than a size cutoff~$c$
and rank~2 to jobs larger than~$c$,
but this begs the question: how do we choose~$c$?

In this section, we give guidelines for designing scheduling policies
for the size-aware and size-oblivious LPL settings.
We tackle the following system design questions:
\begin{itemize}
\item
  How many priority levels do we need to rival
  the performance of ideal policies?
\item
  How should we choose the rank cutoffs between levels?
\item
  Which ideal policy works best when adapted to the LPL setting?
\end{itemize}
The remainder of this section addresses these questions for
the size-aware (\cref{sec:lpl:size-aware})
and size-oblivious (\cref{sec:lpl:size-oblivious})
settings.

  \begin{figure}
    \centering
    \begin{tikzpicture}[figure]
  \begin{scope}[shift={(-10,0)}]
    \node at (13/2, 9.5) {\normalsize\strut\textsc{Arbitrary Policy X}};

    \fill[band_a] (0, 0) rectangle (13, 2);
    \fill[band_b] (0, 2) rectangle (13, 4);
    \fill[band_c] (0, 4) rectangle (13, 6);

    \yguide[$c_1$]{0}{2}
    \yguide[$c_2$]{0}{4}
    \yguide[$c_3$]{0}{6}

    \draw[primary]
    (0, 1.5) -- (3, 3.5)
    -- (5, 1) -- (8, 5.5)
    -- (11, 2.5) -- (12.5, 5);

    \axes{13}{6}{$c_0$}{age}{$0$}{rank}
  \end{scope}

  \begin{scope}[shift={(10,0)}]
    \node at (13/2, 9.5) {\normalsize\strut\textsc{LPL-X}};
    \fill[band_a] (0, 0) rectangle (13, 2);
    \fill[band_b] (0, 2) rectangle (13, 4);
    \fill[band_c] (0, 4) rectangle (13, 6);

    \yguide[$1$]{0}{1}
    \yguide[$2$]{0}{3}
    \yguide[$3$]{0}{5}

    \draw[primary, opacity=0.1]
    (0, 1.5) -- (3, 3.5)
    -- (5, 1) -- (8, 5.5)
    -- (11, 2.5) -- (12.5, 5);

    \draw[primary]
    (0, 1) -- (0.75, 1)
    (0.75, 3) -- (4.2, 3)
    (4.2, 1) -- (5 + 2/3, 1)
    (5 + 2/3, 3) -- (7, 3)
    (7, 5) -- (9.5, 5)
    (9.5, 3) -- (11.9, 3)
    (11.9, 5) -- (13, 5);

    \axes{13}{6}{}{age}{$0$}{rank}
  \end{scope}

  \node at (-21.5 + 13/2, 0) {};
  \node at (21.5 + 13/2, 0) {};
\end{tikzpicture}

    \ifempty{}{}{\begin{justify}
      \footnotesize\ignorespaces\unskip
    \end{justify}}
    \caption{\ignorespaces
  Transforming an Arbitrary Policy~X into LPL-X
\unskip}
    \label{fig:x_to_lpl-x}
  \end{figure}
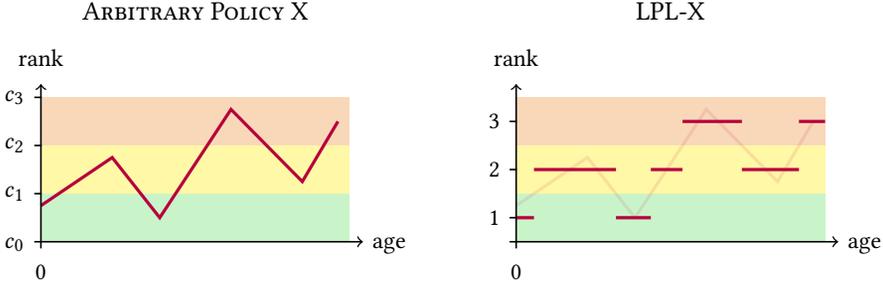

Throughout, we use the following strategy
for adapting scheduling policies to the LPL setting.
Suppose we want to adapt ideal policy~X
which uses a continuum of priority levels.
We create an LPL version of~X, called \emph{LPL-X},
by restricting its rank function to one of finitely many values
in the following way.
Suppose our system allows $n$ priority levels.
We choose a number of \emph{rank cutoffs} $c_1, c_2, \dots, c_{n - 1}$.
Whenever X would assign a job a rank between $c_{i - 1}$ and~$c_i$,
LPL-X assigns the job rank~$i$.\footnote{%
  As edge cases, let $c_0 = 0$ and $c_n = \infty$.}
\Cref{fig:x_to_lpl-x} illustrates this transformation from X to LPL-X.

\subsection{The Size-Aware LPL Setting: LPL-SRPT}
\label{sec:lpl:size-aware}

\begin{table}
  \caption{Highly Variable Job Size Distributions}
  \label{tab:high-var_dist}
  \centering
  \begin{tabular}{@{}lll@{}}
    \toprule
    \textsc{Name} & \textsc{Density Function} & \textsc{Squared Coefficient of Variation} \\
    \midrule
    Bounded Pareto & $f_S(x) \approx 1/x^2$ for $1 \leq x < 100000$ & $C_S^2 \approx 753$ \\[0.2em]
    Weibull & $f_S(x) = x^{-3/4}\exp(-x^{1/4})/4$ for $x \geq 0$ & $C_S^2 = 69$ \\
    \bottomrule
  \end{tabular}
\end{table}

\newcommand{\figLplBarsSrpt}{\figureMathematica{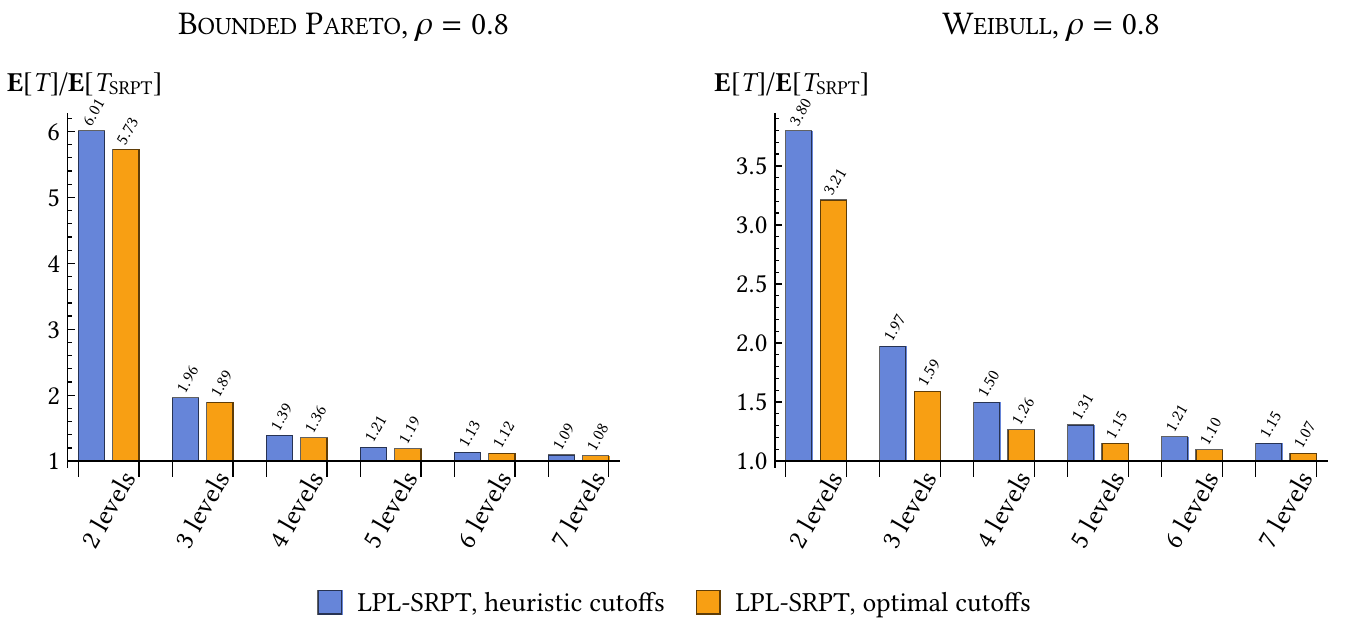}{
  Response Time of LPL-SRPT as a Function of Number of Levels
}{}}

We now address the question of how to minimize mean response time
in the size-aware LPL setting.
Inspired by SRPT's optimality, we investigate LPL-SRPT.
We can think of LPL-SRPT as assigning rank~$i$
to a job with remaining size between cutoffs $c_{i - 1}$ and~$c_i$.

\Cref{fig:lpl_bars_srpt_80} compares LPL-SRPT to SRPT
as a function of the number of priority levels
for two different job size distributions,
which are described in \cref{tab:high-var_dist}.
For brevity, we show only moderate load $\rho = 0.8$,
but the trends are virtually identical at other loads.
Optimizing LPL-SRPT's mean response time amounts to
tuning the cutoffs $c_1, c_2, \dots, c_{n - 1}$.
We consider two strategies for tuning the cutoffs:
\begin{itemize}
\item
  \emph{Heuristic cutoffs:} we set the rank cutoffs
  to split load evenly between the ranks,\footnote{%
    More formally, this means setting the cutoffs such that
    $\E{S \1(c_{i - 1} \leq S < c_i)}$ has the same value
    for all $i \in \{1, 2, \dots, n\}$,
    where $\1$ is the indicator function.}
  a simple heuristic which has been used in practice
  \citep{montazeri_homa_2018, harchol-balter_size-based_2003}.
  The resulting cutoffs depend on the size distribution~$S$
  but not on the load~$\rho$.
\item
  \emph{Optimal cutoffs:} we numerically optimize the rank cutoffs
  to yield minimal mean response time.
  The resulting cutoffs depend on both the job size distribution~$S$
  and the load~$\rho$.
\end{itemize}

\figLplBarsSrpt

The job size distributions in \cref{tab:high-var_dist}
are both high-variance job size distributions.
We expect our conclusions generalize
to other high-variance job size distributions.
The low-variance case turns out to be less interesting,
because as mentioned in \cref{sec:substitute:size-aware:nonpreemptive},
simple policies like FCFS already offer good performance.

\subsubsection{How many priority levels do we need to rival the performance of ideal policies?}
\label{sec:lpl:size-aware:how_many}

We see that \emph{roughly 6~priority levels is enough}
to approach the mean response time of SRPT,
the optimal ideal policy in the size-aware setting.
Even with the heuristic cutoffs,
6~levels gives mean response time within 21\% of SRPT's
at the moderate load $\rho = 0.8$.
This supports the empirical findings of LPL system designers:
\citet{harchol-balter_size-based_2003} use 6~levels,
and \citet{montazeri_homa_2018} use up to 7~levels.

With this said, we note that using just 2~priority levels
is still an order-of-magnitude improvement over FCFS.
This is because FCFS has very poor mean response time
for the high-variance job size distributions,
because small and medium jobs can get stuck behind very long jobs.
\citep{harchol-balter_performance_2013}.

\subsubsection{How should we choose the rank cutoffs between levels?}
\label{sec:lpl:size-aware:cutoffs}


We see that the
\emph{load-balancing heuristic cutoffs serve as a good rule of thumb}.
In additional numerical experiments, omitted for brevity,
we tried several other heuristics,
such as evenly splitting the number of arrivals in each bucket
or using a geometric sequence for the rank cutoffs~$c_i$.
None of these alternatives performed well as consistently
as the load-balancing heuristic used in \cref{fig:lpl_bars_srpt_80}.

Additional examples at other loads,
omitted for brevity,
show that the gap between the heuristic and optimal cutoffs
becomes more important under higher load,
especially for small numbers of priority levels.
However, using optimal cutoffs presents a practical difficulty:
the optimal cutoffs depend on the load~$\rho$.
Fortunately, if we optimize the cutoffs for $\rho = 0.8$,
additional numerical experiments, omitted for brevity,
show that we achieve near-optimal performance for all but the highest loads.

\subsubsection{Which ideal policy works best when adapted to the LPL setting?}
\label{sec:lpl:size-aware:which_ideal}

\newcommand{\figLplBarsPsjf}{\figureMathematica{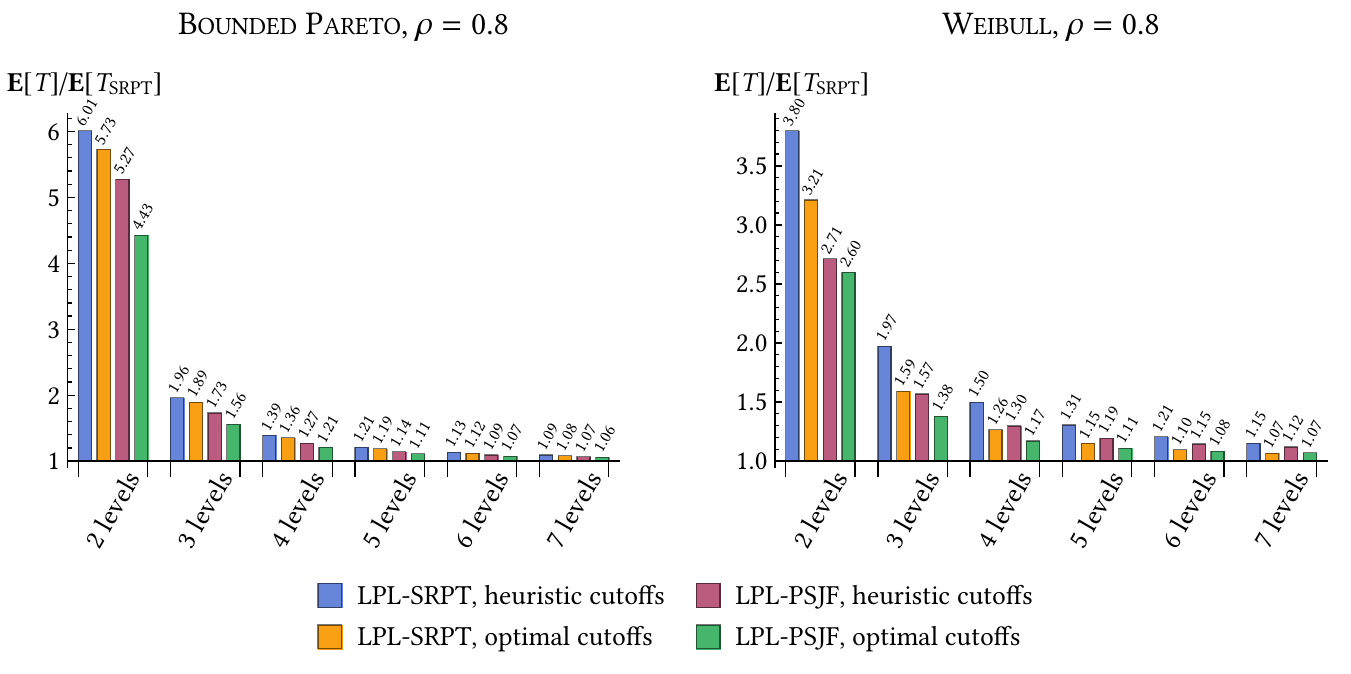}{
  Response Time Comparison: LPL-SRPT vs. LPL-PSJF
}{}}

\figLplBarsPsjf

We began this section by focusing on LPL-SRPT.
However, even though SRPT is the optimal policy with infinite priority levels,
LPL-SRPT turns out \emph{not} to be the best policy in the LPL setting.

Recall that SRPT uses a job's \emph{remaining} size as its rank
(\cref{fig:rank_examples}).
This means that LPL-SRPT ``upgrades'' jobs to the next priority level
when they become small enough.
Unfortunately, because of the limited number of priority levels,
this can cause smaller (new) jobs to have to wait behind larger (recently upgraded) jobs.
Fortunately, it turns out that a simple alteration to LPL-SRPT
can avoid this issue:
never change a job's rank.
This results in a policy called LPL-PSJF.\footnote{%
  One can view LPL-PSJF as a version of the P-Prio policy.}
The name comes from the \emph{Preemptive Shortest Job First} (PSJF) policy,
under which a job's rank is its \emph{original} size,
as opposed to its remaining size.
\Cref{fig:lpl_bars_psjf_80} shows that, counter to intuition,
LPL-PSJF outperforms LPL-SRPT.


\subsection{Designing an LPL Policy for the Size-Oblivious Setting}
\label{sec:lpl:size-oblivious}

\newcommand{\figLplBarsFb}{\figureMathematica{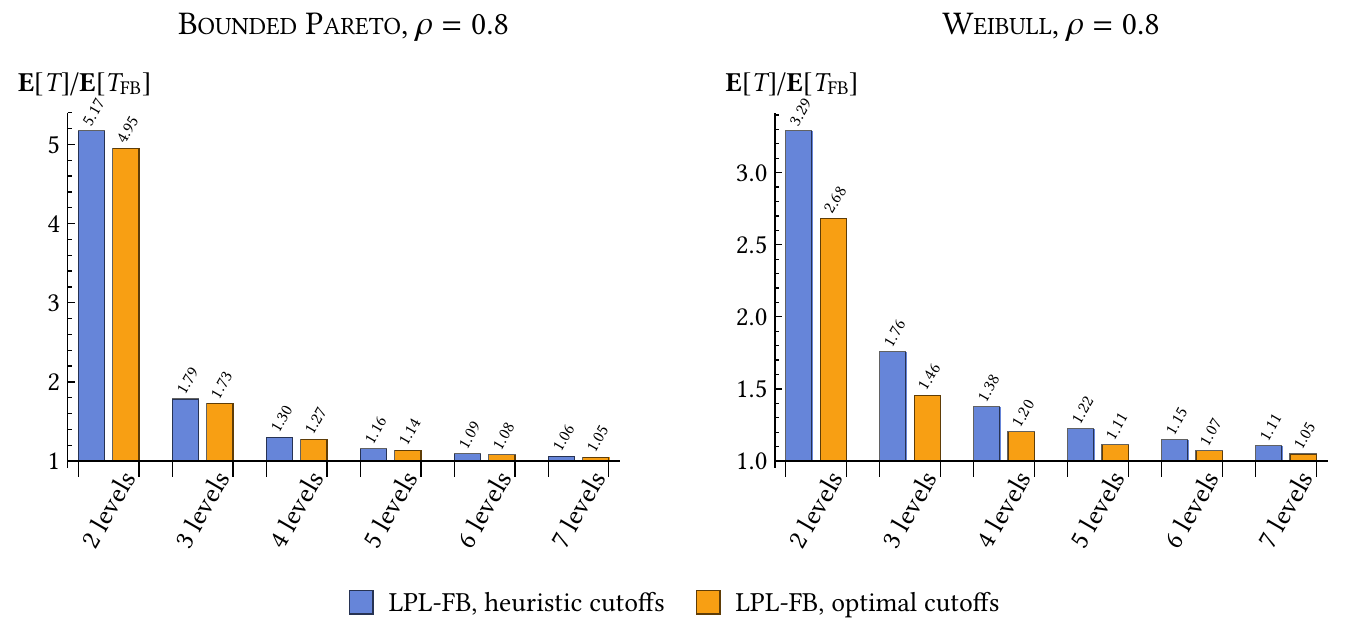}{
  Response Time of LPL-FB as a Function of Number of Levels
}{}}

We now turn to the size-oblivious LPL setting.
For low-variance job size distributions,
we have seen that FCFS can serve as a good Gittins substitute
(\cref{sec:substitute:size-oblivious:further_simplify}),
so we focus on the high-variance job size distributions
from \cref{tab:high-var_dist}.
For these distributions, FB is either optimal or a good Gittins substitute
(\cref{sec:substitute:size-oblivious:further_simplify}),
so LPL-FB is a promising policy for the LPL setting.
We can think of LPL-FB as assigning rank~$i$
to a job with age between cutoffs $c_{i - 1}$ and~$c_i$.

\Cref{fig:lpl_bars_fb_80} compares LPL-FB to FB
as a function of the number of priority levels
for two different job size distributions,
which are described in \cref{tab:high-var_dist}.
We consider the same two strategies for tuning the cutoffs as in
\cref{sec:lpl:size-aware}:
load-balancing \emph{heuristic cutoffs}
and numerically-solved \emph{optimal cutoffs}.

\figLplBarsFb

\subsubsection{How many priority levels do we need to rival the performance of ideal policies?}
\label{sec:lpl:size-oblivious:how_many}

We reach a conclusion similar to the size-aware setting:
\emph{roughly 5~priority levels is enough}
to approach the mean response time of FB,
the optimal or near-optimal ideal policy in the size-oblivious setting
for the job size distributions in \cref{tab:high-var_dist}.
Even with the heuristic cutoffs,
5~levels gives mean response time within 23\% of FB's
at the moderate load $\rho = 0.8$.

With this said,
as in the size-aware setting (\cref{sec:lpl:size-aware:how_many}),
using just 2~priority levels is still an order-of-magnitude improvement over using FCFS.

\subsubsection{How should we choose the rank cutoffs between levels?}
\label{sec:lpl:size-oblivious:cutoffs}

We again reach a conclusion similar to the size-aware setting:
\emph{load-balancing heuristic cutoffs serve as a good rule of thumb}.
In particular, this heuristic outperformed
the other heuristics we tried (\cref{sec:lpl:size-aware:cutoffs}).
Using optimal cutoffs again becomes more important at higher load.

\subsubsection{Which ideal policy works best when adapted to the LPL setting?}
\label{sec:lpl:size-oblivious:which_ideal}

We have not found an LPL policy that significantly improves upon LPL-FB
in the size-oblivious setting.
There is some theoretical evidence supporting this conclusion.
\Citet{scully_simple_2020} propose a policy which in some ways resembles LPL-FB,
which they prove always achieves mean response time within a factor of~5 of Gittins's.

In some ways, the conclusion that FB is the best policy to adapt
to the size-oblivious LPL setting
matches our conclusion for the size-aware LPL setting,
in which we found that LPL-PSJF outperformed LPL-SRPT.
Both LPL-FB and LPL-PSJF have the property that
\emph{a job's rank never improves}.
This property also holds for the policy proposed by \citet{scully_simple_2020}.


\subsection{Summary: 5 or 6 Levels Is Enough}

We have seen that one can obtain near-optimal mean response time in the LPL setting
with just 5 or 6 priority levels,
using LPL-SRPT in the size-aware setting
and LPL-FB in the size-oblivious setting.
Even just 2~priority levels gives an order-of-magnitude improvement over FCFS.
A simple load-balancing heuristic suffices
for choosing the rank cutoffs between priority levels,
though it is possible to improve upon this already-good heuristic
(\cref{sec:lpl:size-aware:which_ideal}).

\section{Preemption Checkpoints}
\label{sec:pc}

\newcommand{\figPcFb}{%
  \begin{figure}
    \centering
    \begin{tikzpicture}[figure]
  \begin{scope}[shift={(-10,0)}]
    \node at (6/2, 11.5) {\normalsize\strut\textsc{Without Checkpoints}};

    \draw[primary]
    (0, 0) -- (6, 6);

    \axes{6}{6}{$0$}{age}{$0$}{rank}
  \end{scope}

  \begin{scope}[shift={(10,0)}]
    \node at (6/2, 11.5) {\normalsize\strut\textsc{With Checkpoints}};

    \xguide[$\delta$]{1.75}{3.75}
    \xguide[$2\delta$]{3.5}{5.5}
    \xguide[$3\delta$]{5.25}{7.25}

    \draw[primary, opacity=0.1]
    (0, 2) -- (6, 8);

    \draw[primary]
    (0, 1) -- (6, 1);

    \axes{6}{8}{$0$}{age}{$0$}{rank}

    \jump{0}{1}{2}
    \jump{1.75}{1}{3.75}
    \jump{3.5}{1}{5.5}
    \jump{5.25}{1}{7.25}
  \end{scope}

  \node at (-21.5 + 6/2, 0) {};
  \node at (21.5 + 6/2, 0) {};
\end{tikzpicture}

    \ifempty{}{}{\begin{justify}
      \footnotesize\ignorespaces\unskip
    \end{justify}}
    \caption{\ignorespaces
  How Preemption Checkpoints Affect the Rank Function of FB
\unskip}
    \label{fig:pc_fb}
  \end{figure}
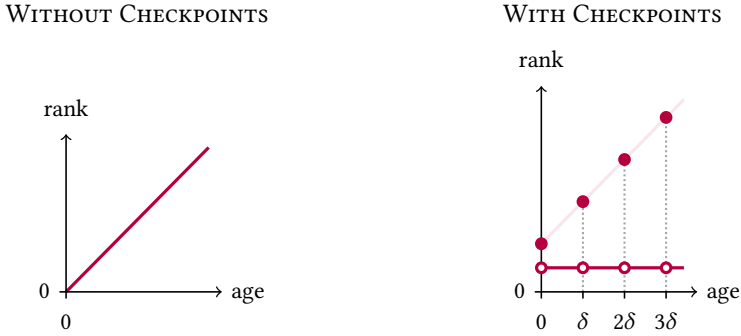}

\ifmanuscript{\figPcFb}{}

Thus far we have been assuming that jobs may be freely preempted,
but this is far from the case in practical computer systems.
Programs can have temporary state at all levels of the memory hierarchy,
from registers to RAM.
For cloud-based jobs, even the disk may be temporary state.
Before preempting a job, we must either save or discard its state.

Checkpointing is a common solution used to combat lost state.
Every job occasionally saves its transient state,
and we allow preemptions only immediately after saving.
We refer to ages when a job saves its work as \emph{checkpoints}.
Because saving work takes time, each checkpoint adds to a job's size.

A very similar situation occurs when scheduling packet flows in networks,
e.g. at a network switch.
In this setting, each packet flow is a job,
and serving a job corresponds to sending its packets.
Packets are indivisible, and each packet incurs overhead in the form of its header.
We can think of packet boundaries as analogous to checkpoints.

When scheduling in a system with checkpointing,
the key question is: \emph{how much time should there be between checkpoints?}
Or, in the context of packet flows: \emph{how large should packets be?}
Answering this question requires balancing a delicate tradeoff.
\begin{itemize}
\item
  On one hand, less time between checkpoints allows for quicker preemptions.
  This decreases mean response time
  because small jobs are less likely to get stuck during
  long uninterruptible periods between checkpoints.
\item
  On the other hand, checkpointing takes time,
  so the less time between checkpoints add more to the system load.
  This in turn can increase mean response time or even cause instability.
\end{itemize}

We give a rule of thumb for balancing this tradeoff.
After discussing in more detail how we model checkpoints (\cref{sec:pc:model}),
we determine how frequently checkpoints should occur
to minimize mean response time (\cref{sec:pc:tradeoff}).
Throughout, we focus on the size-oblivious setting,
but in additional numerical experiments, omitted for brevity,
we have observed the same results in the size-aware setting.

\subsection{Jobs with Preemption Checkpoints}
\label{sec:pc:model}

\ifmanuscript{}{\figPcFb}

We consider a system where jobs cannot be preempted unless their work is saved.
Jobs save their work at specific ages called \emph{checkpoints}.
Upon reaching a checkpoint,
a job takes a deterministic amount of \emph{overhead time}~$\gamma$ to save its work,
after which the job may be preempted.
The overhead time does not count towards a job's age.

A job can only be preempted at checkpoint ages.
Given a scheduling policy,
we can express this constraint by modifying the scheduling policy's rank function.
Specifically, we leave the rank function as-is for checkpoint ages,
but for intermediate ages, we set the rank to a low value.
This gives a job between checkpoints priority over jobs at checkpoints.
\Cref{fig:pc_fb} shows what this looks like for FB.

For simplicity, we study the case where checkpoint ages
are distributed evenly with \emph{age gap}~$\delta$.
That is, checkpoints occur at ages $\delta, 2\delta, 3\delta, \dots.$
Our main task is to optimize the gap~$\delta$
given the job size distribution~$S$, load~$\rho$, and overhead~$\gamma$.

\subsection{Optimizing the Gap Between Checkpoints}
\label{sec:pc:tradeoff}

\newcommand{\figPcLine}{\figureMathematica{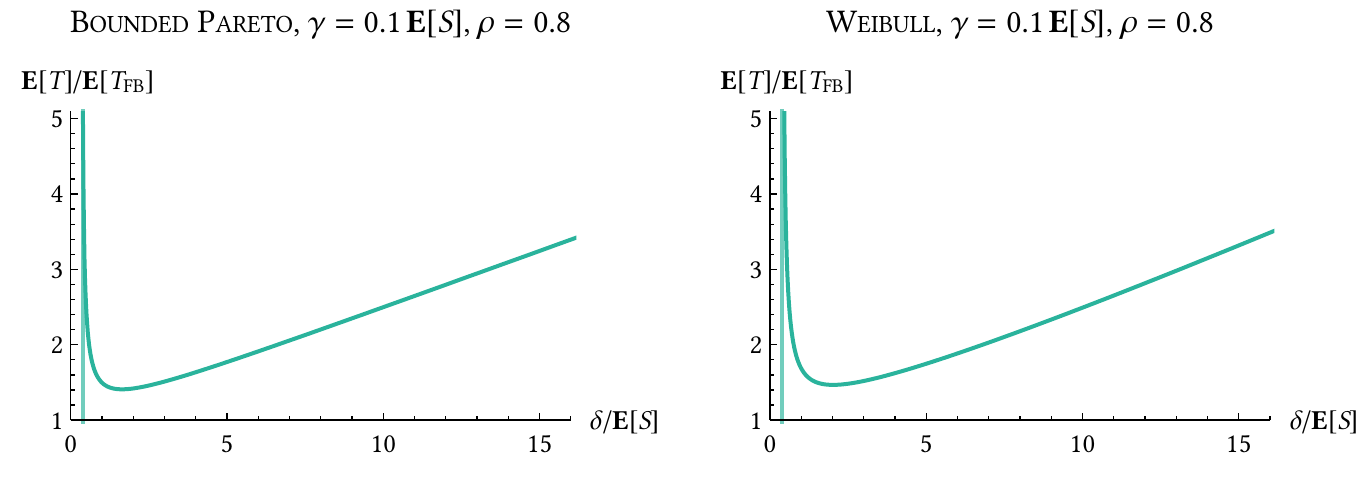}{
  Response Time as a Function of Checkpoint Interval~$\delta$, Large Overhead $\gamma = 0.1\,\E{S}$, Load $\rho = 0.8$
}{
  We show mean response time of FB with checkpoints
  as a function of the normalized checkpoint interval~$\delta/\E{S}$
  for a system with large checkpoint overhead, namely $\gamma = 0.1\,\E{S}$.
  We use FB in a system \emph{without} checkpoints or overhead as a baseline for comparison.
  We consider Bounded Pareto and Weibull job size distributions
  (\cref{tab:high-var_dist})
  and load $\rho = 0.8$.
  As we decrease~$\delta$, mean response time decreases approximately linearly
  until a critical value ``left wall'' value of~$\delta$,
  at which point the system becomes unstable and mean response time approaches infinity.
  The vertical lines show a conservative estimate, given by~\cref{eq:delta_safe},
  of where the left wall is.
}}

\newcommand{\figPcLargeOverhead}{\figureMathematica{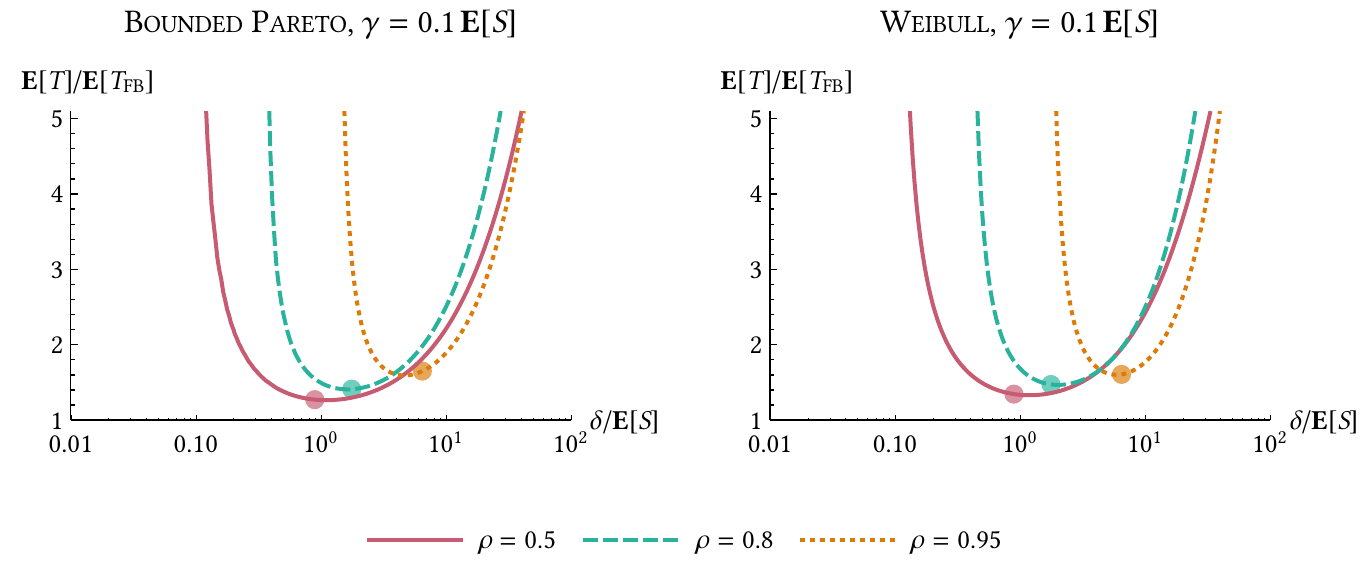}{
  Response Time as a Function of Checkpoint Interval~$\delta$ (Log Scale), Large Overhead $\gamma = 0.1\,\E{S}$
}{
  We show normalized mean response time as a function of the normalized checkpoint interval~$\delta/\E{S}$
  for a system with large checkpoint overhead, namely $\gamma = 0.1\,\E{S}$.
  We consider Bounded Pareto and Weibull job size distributions
  (\cref{tab:high-var_dist})
  for several values of load~$\rho$.
  We highlight the rule-of-thumb value of $\delta$ given by \cref{eq:pc_rule_of_thumb}.
}}

\newcommand{\figPcSmallOverhead}{\figureMathematica{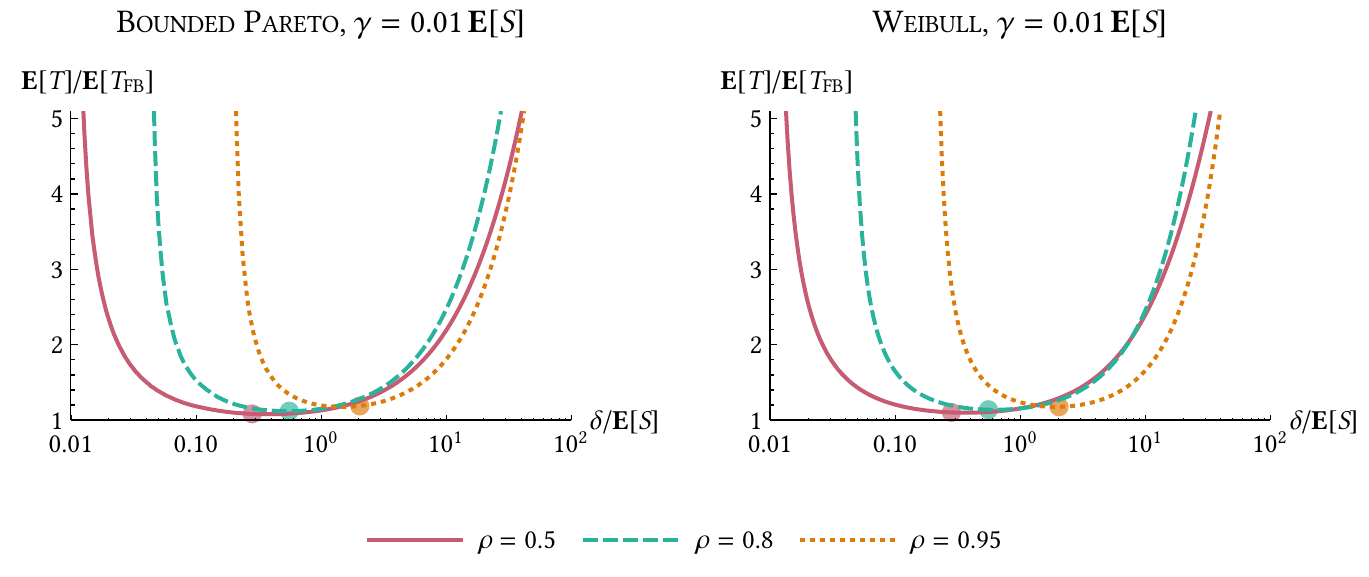}{
  Response Time as a Function of Checkpoint Interval~$\delta$ (Log Scale), Small Overhead $\gamma = 0.01\,\E{S}$
}{
  We show normalized mean response time as a function of the normalized checkpoint interval~$\delta/\E{S}$
  for a system with small checkpoint overhead, namely $\gamma = 0.01\,\E{S}$.
  We consider Bounded Pareto and Weibull job size distributions
  (\cref{tab:high-var_dist})
  for several values of load~$\rho$.
  We highlight the rule of thumb for the optimal value of $\delta$ given by \cref{eq:pc_rule_of_thumb}.
}}

\ifmanuscript{\figPcLine}{}

We focus on the size-oblivious setting
with the high-variance job size distributions from \cref{tab:high-var_dist}.\footnote{%
  To make the analysis numerically tractable,
  we actually truncate the distributions at size~$5000$
  and discretize them into increments of size~$0.125$.
  The trends we observe are not sensitive to these values.}
For these job size distributions, FB is a near-optimal Gittins substitute,
so we schedule using the version of FB with checkpoints shown in \cref{fig:pc_fb}.

In \cref{fig:pc_line}
we show how mean response time varies as a function of the gap~$\delta$ between checkpoints.
The figure shows the specific case of relatively large overhead $\gamma = 0.1\,\E{S}$ at load $\rho = 0.8$.
As we will show later,
the trends are similar for small overhead and other loads
(\cref{fig:pc_small_overhead, fig:pc_large_overhead}).

\ifmanuscript{}{\figPcLine}

For large enough~$\delta$,
\cref{fig:pc_line} shows that mean response time is a roughly linear function of~$\delta$,
where smaller $\delta$ is better.
However, there is a vertical asymptote as $\delta$ approaches a small value,
which we call the ``left wall''.
This leaves us in a precarious situation:
smaller values of~$\delta$ generally yield lower mean response time,
but if $\delta$ becomes too small,
mean response time suddenly becomes infinite.
In the remainder of this section, we explain how to choose a value of~$\delta$
that is small enough but not too small.

\subsubsection{Ensuring stability with a safe checkpoint age gap}

As a first step towards optimizing~$\delta$,
we must figure out the value $\delta_{\text{left wall}}$
at which the left wall asymptote occurs.
The asymptote at the left wall is caused by checkpoints becoming so frequent
that the overheads make the system unstable,
which sends mean response time to infinity.

It turns out that $\delta_{\text{left wall}}$ has a complicated formula,
but there is a simple formula that bounds it.
Let
\begin{align}
  \label{eq:delta_safe}
  \delta_{\text{safe}} = \frac{\gamma\rho}{1 - \rho}.
\end{align}
As we prove in \cref{sec:delta_safe},
we always have $\delta_{\text{left wall}} \leq \delta_{\text{safe}}$.
This means the system is guaranteed to be stable,
and thus the mean response time is finite,
whenever $\delta > \delta_{\text{safe}}$.
This holds not just for this example
but under any scheduling policy and any job size distribution.

\ifmanuscript{\figPcLargeOverhead}{}

\subsubsection{Optimizing the checkpoint age gap}

We want to choose a gap~$\delta$ between checkpoints
that is not just stable but also optimizes mean response time.
We see from \cref{fig:pc_line} that
the optimal value of $\delta$ is slightly larger than $\delta_{\text{safe}}$,
but it is not clear how much larger is optimal.
Through a number of examples,
we have found the following rule of thumb to give near-optimal performance
in the typical case
when the overhead~$\gamma$ is a fraction of the mean job size~$\E{S}$:
\begin{align}
  \label{eq:pc_rule_of_thumb}
  \delta_{\text{rule of thumb}} = \frac{1}{1 - \rho} \sqrt{\frac{\gamma \E{S}}{\rho}}.
\end{align}
\Cref{fig:pc_large_overhead, fig:pc_small_overhead} show that
this rule-of-thumb gap yields near-optimal performance
for a range of loads~$\rho$ and overheads~$\gamma$.

\ifmanuscript{}{\figPcLargeOverhead}

The intuition behind \cref{eq:pc_rule_of_thumb} is as follows.
We see in \cref{fig:pc_large_overhead, fig:pc_small_overhead} that
mean response time has a ``bathtub'' shape
when plotted as a function of~$\delta$ on a log scale.
Each bathtub is approximately symmetrical near its minimum.
This suggests that if we can find the values of~$\delta$
corresponding to the ``walls'' of the bathtub,
then their geometric mean would be a good rule of thumb.
This is exactly the approach we take.
\begin{itemize}
\item
  \emph{Left wall:} we use $\delta_{\text{safe}}$
  as an approximation for the left wall.
\item
  \emph{Right wall:} through additional numerical experiments,
  omitted for brevity,
  we have found that the formula $\E{S}/(\rho^2 (1 - \rho))$
  is a good approximation for the right wall.
  Our search for a good approximation was guided by the observation
  that the right wall is virtually unaffected by the overhead~$\gamma$,
  as we see by comparing \cref{fig:pc_large_overhead, fig:pc_small_overhead}.
  This makes sense: when the gap~$\delta$ between checkpoints is large,
  there are very few preemptions,
  so the overhead~$\gamma$ has little impact of mean response time.
\end{itemize}

\figPcSmallOverhead

\subsection{Summary: Rule-Of-Thumb Formula}

When scheduling in systems with checkpointing,
there is a tradeoff between making checkpoints less frequent,
which avoids checkpoint-associated overhead,
and more frequent,
which enables smarter scheduling.
In \cref{eq:pc_rule_of_thumb},
we propose a rule-of-thumb formula
that yields a near-optimal gap between checkpoints.
Although we focus throughout on the size-oblivious setting,
we have observed that this rule of thumb
also works well in the size-aware setting.


\section{Prior Work}
\label{sec:prior_work}

The purpose of this paper is not to prove new results in queueing theory.
Instead, our goal is to draw upon new developments in queueing theory
and extract actionable lessons for computer system design.
In this prior work section, we walk through each section of the paper,
explaining what was previously known and what we contribute.
Broadly speaking,
while prior work lays theoretical foundations for analyzing
all of the policies in this paper,
we are the first to apply that theory to synthesizing practical takeaways
for the scenarios we study.

\subsection{Prior Work on Creating a Gittins Substitute (\cref{sec:substitute})}

All of the scheduling policies discussed in \cref{sec:substitute}
are well-known policies from the scheduling literature.
The SRPT and Gittins policies
are known to minimize mean response time in their respective settings
\citep{schrage_proof_1968, gittins_multi-armed_2011}.
However, despite the optimality of Gittins being known,
it was until recently not known how to compute
the mean response time achieved by Gittins.
This changed in \citeyear{scully_soap_2018},
when \citet{scully_soap_2018} introduced the SOAP framework,
which gives a technique for analyzing the mean response times
of Gittins, SERPT, and numerous other policies.

In principle, the SOAP framework makes it possible to study the question of
whether SERPT is a good Gittins substitute.
However, to the best of our knowledge,
only one prior study attempts to address this question.
\Citet{scully_simple_2020} prove that in the size-oblivious setting,
a modification of SERPT achieves within a factor of~5 of Gittins's mean response time
for \emph{any} job size distribution~$S$.
Unfortunately,
while this factor-of-5 bound is an important step from a theoretical perspective,
it is too loose to be useful in practice.
Moreover, the result applies only in the size-oblivious setting,
leaving the case of the class-aware setting open.

\subsubsection{Our contribution}

We give the first study that compares Gittins with SERPT
from a practical perspective.
Rather than attempting to prove a universal result,
we compare Gittins and SERPT in a wide variety of concrete examples,
covering both size-oblivious and class-aware settings.
To the best of our knowledge,
ours is the first comparison of this kind,
and findings like those shown in \cref{fig:worst} are new.
We find that SERPT consistently achieves much better mean response time
than the best known theoretical bound would suggest
\citep{scully_simple_2020}.
We thus believe that our example-driven approach
gives the clearest picture yet of how SERPT compares to Gittins in practice.

\subsection{Prior Work on Limited Priority Levels (\cref{sec:lpl})}

The specific LPL policies we study in \cref{sec:lpl}
have been studied in both systems and theory,
although they do not have consistently used names.

In systems, both \citet{harchol-balter_size-based_2003} and \citet{montazeri_homa_2018}
design scheduling policies for size-aware computer systems
with limited priority levels,
arriving at variants of LPL-SRPT and LPL-PSJF, respectively.

In theory, \citet{kleinrock_queueing_1976} introduces and analyzes
a broad class of scheduling policies that includes LPL-FB.
Similarly, the previously discussed SOAP framework \citep{scully_soap_2018}
can be used to analyze the mean response time of LPL-SRPT, LPL-PSJF, and LPL-FB.
Most recently, \citet{chen_scheduling_2020} prove theorems about
LPL-SRPT with 2 priority levels
in the heavy-traffic limit, meaning the $\rho \to 1$ limit.
The main result is that even with just 2 levels,
LPL-SRPT has asymptotically better mean response time than FCFS,
provided one chooses the cutoff carefully as a function of load~$\rho$.

\subsubsection{Our contribution}

We view our work as a middle ground between
the prior work in systems and the prior work in theory.
Compared to the systems,
we do a more mathematically rigorous analysis.
Compared to the theory,
we study on a more practical question,
considering up to 7 priority levels and all loads.\footnote{%
  For brevity, the figures in \cref{sec:lpl} show $\rho = 0.8$,
  but the trends are virtually identical at other loads.}
Our findings provide validation
for design choices that have been used in systems
\citep{harchol-balter_size-based_2003, montazeri_homa_2018},
such as using a load-balancing heuristic for choosing
the remaining size cutoffs between levels.

Although the LPL-SRPT and LPL-PSJF policies have been used in the past,
we believe that viewing them as special cases of
a generic ``X to LPL-X'' transformation (\cref{fig:x_to_lpl-x})
is a novel perspective.
For instance, an interesting future direction might be to study
LPL-Gittins or LPL-SERPT.

\subsection{Prior Work on Preemption Checkpoints (\cref{sec:pc})}

Checkpointing has been well studied in queueing systems
from the perspective of using checkpoints to minimize work lost
in the event of system crashes
\citep{nicola_comparative_1990, dohi_optimal_2000, aupy_checkpointing_2016}.
However, this line of work assumes FCFS scheduling.
In contrast, our goal is to minimize mean response time in systems with checkpointing,
which requires considering other scheduling policies.

One of the only works incorporating both checkpointing and non-FCFS scheduling
is that of \citet{goerg_further_1990},
which analyzes size-aware systems under a variant of SRPT with checkpointing.
However, it stops short of characterizing the optimal gap between checkpoints.
In principle, one could use the SOAP framework \citep{scully_soap_2018}
to compute the mean response time of other policies in systems with checkpointing.

\subsubsection{Our contribution}

We provide the first study of checkpointing with non-FCFS scheduling,
namely a variation of FB,
in the size-oblivious setting.
We use the SOAP framework in our mean response time computations.
Our main result is a simple rule of thumb
that yields a near-optimal gap between checkpoints.\footnote{%
  For brevity, the figures in \cref{sec:pc} focus on the size-oblivious setting,
  but we have observed that the rule of thumb works well
  in the size-aware setting, too.}

\section{Lessons Learned}
\label{sec:lessons}

Scheduling has a profound impact on response time in queueing systems.
Unfortunately, the policies that minimize mean response time in theory,
namely SRPT and Gittins,
can be difficult or impossible to implement perfectly in practice.
In this work, we have shown how to achieve near-optimal performance
under real-world constraints.
Below, we summarize the main lessons learned.
\begin{itemize}
\item
  When job sizes are unknown or only partially known to the scheduler,
  the Gittins policy can be impractically complicated to implement.
  We show that \emph{the simpler SERPT policy serves as a good substitute for Gittins}
  in both the size-oblivious and class-aware settings
  (\cref{sec:substitute}).
\item
  Some systems, such as network switches,
  restrict the scheduler to a limited number of priority levels,
  making even simple policies like SRPT and FB impossible to implement.
  We show that \emph{just 5 or 6 priority levels suffices}
  to achieve near-optimal mean response time,
  with even just 2 levels giving an order-of-magnitude improvement over FCFS
  (\cref{sec:lpl}).
\item
  In some settings,
  one can only safely preempt a job if it has saved its work at a checkpoint.
  One instance of this is scheduling packet flows in a network,
  where checkpoints occur after sending each packet.
  While more frequent checkpoints allows for smarter scheduling,
  each such checkpoint incurs some overhead.
  Too many checkpoints can cause the system to become unstable,
  sending mean response time to infinity.
  Finding the optimal checkpoint frequency thus requires balancing a delicate tradeoff.
  We give a simple rule-of-thumb formula for checkpoint frequency
  that yields near-optimal performance
  (\cref{sec:pc}).
\end{itemize}

\begin{acks}
  We thank Srinivasan Seshan, Andy Pavlo, and Michael Kuchnik
  for their detailed comments on an earlier draft of this work.
  We also thank John Ousterhout, Behnam Montazeri,
  Douglas G. Down, Maryam Akbari-Moghaddam,
  and Isaac Grosof
  for helpful discussions.

  This work was supported by \grantsponsor{nsf}{NSF}{https://www.nsf.gov} grants
  \grantnum{nsf}{CMMI-1938909}, \grantnum{nsf}{XPS-1629444}, and \grantnum{nsf}{CSR-1763701};
  and a \grantsponsor{google}{Google}{https://research.google/outreach/} 2020 Faculty Research Award.
\end{acks}

\bibliographystyle{ACM-Reference-Format}
\bibliography{refs}

\appendix

\section{Performance of FCFS with Bounded Job Sizes}
\label{sec:fcfs_approximation}

\begin{theorem}
  Consider an M/G/1 whose job sizes all lie within an interval $[s_{\min}, s_{\max}]$.
  Then the mean response time of FCFS is at most a factor of $s_{\max}/s_{\min}\esub$
  worse than that of SRPT:
  \begin{align*}
    \E{T_{\text{FCFS}}} \leq \frac{s_{\max}}{s_{\min}} \E{T_{\text{SRPT}}}.
  \end{align*}
\end{theorem}

\begin{proof}
  Little's law \citep{little_littles_2011},
  relates the mean number of jobs in the system $\E{N}$
  to mean response time $\E{T}$ and the arrival rate~$\lambda$:
  \begin{align*}
    \E{N} = \lambda \E{T}.
  \end{align*}
  It thus suffices to show
  \begin{align}
    \label{eq:N_comparison}
    \E{N_{\text{FCFS}}} \leq \frac{s_{\max}}{s_{\min}} \E{N_{\text{SRPT}}}.
  \end{align}

  To prove~\cref{eq:N_comparison},
  we further split the number of jobs in the system~$N$
  into two parts:
  \begin{itemize}
  \item
    the number of jobs in the queue, which we write as~$L$, and
  \item
    the number of jobs in service, which we write as~$M$.
  \end{itemize}
  In an M/G/1,
  the expected number of jobs in service is the fraction of time the system is busy,
  which is $\rho$ under either scheduling policy:
  \begin{align*}
    \E{M_{\text{FCFS}}} = \E{M_{\text{SRPT}}} = \rho.
  \end{align*}
  Because $N = L + M$, it suffices to show
  \begin{align}
    \label{eq:L_comparison}
    \E{L_{\text{FCFS}}} \leq \frac{s_{\max}}{s_{\min}} \E{L_{\text{SRPT}}}.
  \end{align}

  In order to show~\cref{eq:L_comparison},
  we introduce another new concept:
  the \emph{work}, or total remaining size, of a group of jobs.
  We write
  \begin{itemize}
  \item
    $U$ for the work of jobs in the queue,
  \item
    $V$ for the work of jobs currently in service,\footnote{%
      Note that there are either 0 or 1 jobs in service at any time.} and
  \item
    $W = U + V$ for the total work of all jobs.
  \end{itemize}
  It is simple to show that $\E{V}$ has the same value for all scheduling policies
  \citep[Chapter~23]{harchol-balter_performance_2013},
  so $\E{V_{\text{FCFS}}} = \E{V_{\text{SRPT}}}$.
  Furthermore, $\E{W}$ has the same value for all scheduling policies
  that never leave the server idle while there are still jobs in the system,
  so $\E{W_{\text{FCFS}}} = \E{W_{\text{SRPT}}}$.
  Together, these two facts imply
  \begin{align}
    \label{eq:U_comparison}
    \E{U_{\text{FCFS}}} = \E{U_{\text{SRPT}}}.
  \end{align}

  The final step is to use \cref{eq:U_comparison} to prove~\cref{eq:L_comparison}.
  Because all jobs in the queue under FCFS have not started service yet,
  each has remaining size at least~$s_{\min}$, so
  \begin{align*}
    \E{U_{\text{FCFS}}} \geq s_{\min} \E{L_{\text{FCFS}}}.
  \end{align*}
  Similarly, every job has remaining size at most~$s_{\max}$, so
  \begin{align*}
    \E{U_{\text{SRPT}}} \leq s_{\max} \E{L_{\text{SRPT}}}.
  \end{align*}
  Combining \cref{eq:U_comparison} with the above inequalities
  yields~\cref{eq:L_comparison}, as desired.
\end{proof}

\section{Guaranteeing Stability with a Safe Checkpoint Age Gap}
\label{sec:delta_safe}

See \cref{sec:pc:model} for a full description of
the M/G/1 with preemption checkpoints.

\begin{theorem}
  Consider an M/G/1 with preemption checkpoints
  with overhead time~$\gamma$ per checkpoint
  and age gap~$\delta$ between checkpoints.
  Let
  \begin{align*}
    \delta_{\text{safe}} = \frac{\gamma\rho}{1 - \rho}.
  \end{align*}
  If $\delta > \delta_{\text{safe}}\esub$,
  then the system is stable.
\end{theorem}

\begin{proof}
  Let $S$ be the job size distribution
  before accounting for checkpoint overhead.
  The system with checkpoints is effectively
  an M/G/1 with a different job size distribution~$S'$,
  defined as
  \begin{align*}
    S' = S + \floor*{\frac{S}{\delta}} \gamma.
  \end{align*}
  That is, $S'$ accounts for both a job's size
  and all of the checkpoint overhead it incurs.
  The true load of the system after accounting for overhead
  is $\rho' = \lambda \E{S'}$,
  so the system is stable if $\rho' < 1$.
  We compute
  \begin{align*}
    \rho'
    = \lambda \E*{S + \floor*{\frac{S}{\delta}} \gamma}
    \leq \lambda \E*{S + \frac{S}{\delta} \gamma}
    = \rho \gp*{1 + \frac{\gamma}{\delta}},
  \end{align*}
  which is less than~$1$ if $\delta > \delta_{\text{safe}}$.
\end{proof}

\end{document}